\documentclass[a4paper,onecolumn,11pt,unpublished]{quantumarticle}
\pdfoutput=1

\usepackage[utf8]{inputenc}
\usepackage{amsmath}
\usepackage{amssymb}
\usepackage{algorithm}
\usepackage{algpseudocode}
\usepackage{amsthm}
\usepackage{graphicx}
\usepackage{soul}
\usepackage{booktabs}
\usepackage[hidelinks,colorlinks=true]{hyperref}
\usepackage[numbers,sort&compress]{natbib}
\usepackage[normalem]{ulem}
\usepackage{mathtools}
\usepackage{qcircuit}
\usepackage{blkarray}
\usepackage{appendix}
\usepackage{etoc}
\usepackage{multirow}
\usepackage{makecell}

\DeclarePairedDelimiter\bra{\langle}{\rvert}
\DeclarePairedDelimiter\ket{\lvert}{\rangle}
\DeclarePairedDelimiterX\braket[2]{\langle}{\rangle}{#1 \delimsize\vert #2}

\numberwithin{equation}{section}

\newcommand{\calO}{\mathcal{O}}
\newcommand{\calA}{\mathcal{A}}
\newcommand{\bE}{\mathbb{E}}
\DeclareMathOperator{\tr}{Tr}

\newtheorem{theorem}{Theorem}
\newtheorem{lemma}{Lemma}
\newtheorem{proposition}{Proposition}
\newtheorem{definition}{Definition}

\newtheorem{corollary}{Corollary}

\begin{document}
\title{Quantum Semidefinite Programming with Thermal Pure Quantum States}

\author{Oscar Watts}
\email{oscarahwatts@gmail.com}
\affiliation{Quantinuum, Partnership House, Carlisle Place, London SW1P 1BX, United Kingdom}
\author{Yuta Kikuchi}
\email{yuta.kikuchi@quantinuum.com}
\affiliation{Quantinuum K.K., Otemachi Financial City Grand Cube 3F, 1-9-2 Otemachi, Chiyoda-ku, Tokyo, Japan}
\affiliation{Interdisciplinary Theoretical and Mathematical Sciences Program (iTHEMS), RIKEN, Wako, Saitama 351-0198, Japan}
\author{Luuk Coopmans}
\email{luuk.coopmans@quantinuum.com}
\affiliation{Quantinuum, Partnership House, Carlisle Place, London SW1P 1BX, United Kingdom}

\begin{abstract}
Semidefinite programs (SDPs) are a particular class of convex optimization problems with applications in combinatorial optimization,  operational research, and quantum information science. Seminal work by Brand\~{a}o and Svore shows that a ``quantization'' of the matrix multiplicative-weight algorithm can provide approximate solutions to SDPs quadratically faster than the best classical algorithms by using a quantum computer as a Gibbs-state sampler. We propose a modification of this quantum algorithm and show that a similar speedup can be obtained by replacing the Gibbs-state sampler with the preparation of thermal pure quantum (TPQ) states.
While our methodology incurs an additional problem-dependent error, which decreases as the problem size grows, it avoids the preparation of purified Gibbs states, potentially saving a number of ancilla qubits. 
In addition, we identify a spectral condition which, when met, reduces the resources further, and shifts the computational bottleneck from Gibbs state preparation to ground-state energy estimation.
With classical state-vector simulations, we verify the efficiency of the algorithm for particular cases of Hamiltonian learning problems. We are able to obtain approximate solutions for two-dimensional spinless Hubbard and one-dimensional Heisenberg XXZ models for sizes of up to $N=2^{10}$ variables. For the Hubbard model, we provide an estimate of the resource requirements of our algorithm, including the number of Toffoli gates and the number of qubits.       
\end{abstract}

\maketitle

\newpage

\etocdepthtag.toc{mtchapter}
\etocsettagdepth{mtchapter}{subsection}
\etocsettagdepth{mtappendix}{none}


\section{Introduction}

Semidefinite programs (SDPs) are a class of convex optimization problems concerned with optimizing a linear function over the space of positive semidefinite matrices of dimension $N$ subject to $m$ linear constraints. Many industry-relevant use cases can be cast as semidefinite programs, ranging from job-scheduling in operations research~\cite{Skutella2001} to autonomous navigation in robotics~\cite{mangelson2022guaranteed} and quantum channel tomography in quantum information science~\cite{Skrzypczyk2023}.
Approximate solutions can be obtained efficiently, i.e., with polynomial classical computational resources, but as the size of the semidefinite program increases, the need for faster SDP solvers grows.

Quantum algorithms promise improvements in time and space complexity compared to their classical counterparts.
The work by Brand\~{a}o and Svore~\cite{Brandao2017}, and subsequent works~\cite{vanapeldoorn2019, brandao2019, vanApeldoorn2020quantumsdpsolvers}, have proven that MMW-based quantum algorithms using Gibbs-state sampling can obtain approximate solutions to generic SDPs with query complexity $\tilde{\mathcal{O}}\big(\sqrt{m} +\sqrt{N}\big)$.\footnote{To improve readability, in this paper we use the $\tilde{\mathcal{O}}(\cdot)$ notation to hide factors poly-logarithmic in $N$, $m$, and polynomial in other parameters such as the SDP approximation error $\varepsilon$ defined in the main-body of the manuscript. The more precise $\mathcal{O}(\cdot)$ notation will be used later when all the relevant variables have been introduced. See for example Section~\ref{sec:conclusion} for a comparison between the complexity of our and other quantum SDP algorithms.  
}
This is a quadratic speedup over the best possible classical algorithms, which in general have the complexity lower bound of $\tilde{\mathcal{O}}(m+N)$~\cite{Brandao2017}.
In addition, a quantum algorithm based on the interior point method has been developed which exploits a quantum computer's ability to solve systems of linear equations~\cite{kerenidis2018quantum, augustino2022quantum}.
Both these quantum algorithms are expected to require a fault-tolerant quantum computer which is currently not available. Another class of algorithms, known as variational quantum SDP solvers~\cite{patel2021variational, Bharti2022, patti2022quantum}, do not have these limitations and can be run on near-term quantum devices. Yet it is unclear if, and what, speedup this method offers over classical algorithms due to problems such as barren plateaus in the optimization landscape~\cite{McClean2018}.

In this paper, we build upon the quantum SDP algorithm proposed in~\cite{brandao2019} and show that a quantum algorithm using thermal pure quantum (TPQ) states~\cite{Sugiura2012, Powers2023, Coopmans_2023} provides approximate solutions to SDPs with runtime $\tilde{\mathcal{O}}\big(\sqrt{m} +\sqrt{N}\big)$.
We refer to our algorithm as the TPQ-SDP solver.
While our algorithm still requires an error-corrected quantum computer to demonstrate a quantum speedup for generic SDPs, it reduces some of the resource requirements compared to algorithms using purified Gibbs states~\cite{vanapeldoorn2019}.\footnote{To avoid confusion about terminology, note that TPQ states are pure states that are different from Gibbs states and also their purifications. They are from an ensemble of random pure states, which on average approximates Gibbs state expectation values. The precise definition will be given in Definition~\ref{def:TPQ_state} in Section~\ref{sec:notations}.}
We give explicit circuit implementations of most of the quantum operations, leaving only the specific block encoding implementations for the constraint matrices. We rigorously upper bound the number of queries to block-encoding circuits, and lower bound the success probability of preparing TPQ states using the quantum-eigenvalue transformation (QET)~\cite{Gilyen2019, Martyn2021grand}. This bound may be of independent interest. Our construction may also be useful for the development of classical (quantum-inspired) TPQ-SDP solvers.

As TPQ states only approximate Gibbs-state expectation values~\cite{Sugiura2012, Coopmans_2023}, our algorithm incurs another approximation error. This error reduces the success probability of the algorithm, and is expected to decrease when the size $N$ of the SDP is increased. How fast the error decreases depends on the particular problem one tries to solve and the number of constraints $m$. 
We argue that for physical Gibbs states, with an extensive free energy, the error decreases exponentially as $\mathcal{O}(e^{-\alpha \log{N}})$, where $\alpha$ is a constant determined by the problem. For the generic scenario, we prove a spectral condition on the Gibbs-state Hamiltonian which ensures that the error scales with $N^{-(1-2\tilde{\nu})}$, where $0<\tilde{\nu}<1/2$ depends on the problem. This condition reduces the complexity of preparing the TPQ states, and shifts the computational bottleneck from Gibbs-state preparation to ground-state energy estimation.  At the same time, the approximation error may make the current version of our algorithm unsuitable for SDPs where the desired target precision decreases at least linearly with $N$, or where $m\sim N$, such as relaxed QUBO problems like maximum cut~\cite{Goemans1995}. However, many other quantum SDP solvers are also not able to provide a genuine quantum speedup for such problems, which remains an active area of research~\cite{G_S_L_Brand_o_2022, Augustino2023}. 

The remainder of this paper is organized as follows. 
In Section~\ref{sec:notations}, we start by introducing the notations and conventions that are frequently used throughout the paper. 
We formally define the SDP optimization problem and review the classical matrix-multiplicative weight method in Section~\ref{sec:SDP}. 
In Section~\ref{sec:QSDP}, we present our quantum algorithm based on TPQ states and the main theorem for its complexity. 
In Section \ref{sec:numerics}, we show that our SDP solver can be exploited for solving the Hamiltonian learning problem, in particular for geometrically local quantum Hamiltonians such as the Hubbard model and XXZ Heisenberg model. We verify this with classical numerical simulations of our algorithm for problems of up to $N=2^n=1024$ variables and provide a table with resource requirements to realize the speedup on a quantum device. 
Finally, we discuss our findings and provide an outlook in Section~\ref{sec:conclusion}.

 
\section{Notations and Conventions}
\label{sec:notations}

We summarize the notations and conventions used in this paper. 
We let $\ln{y}$ and $\log{y}$ denote the natural and binary logarithm of $y\in\mathbb{C}$, respectively.
For a Hermitian matrix $A$, the minimum and maximum eigenvalues of $A$ are denoted by $\lambda_\mathrm{min}(A)$ and $\lambda_\mathrm{max}(A)$, respectively. 
The spectral norm of $A$ is defined by 
\begin{equation}
    \lVert A\rVert = \max_{\{\ket{x}; \braket{x}{x}=1\}}\bra{x}A\ket{x}.
\end{equation}
The expression $A\succeq B$ means that the matrix $A-B$ is positive semidefinite, i.e., $\lambda_\mathrm{min}(A-B)\ge0$, and $A\succ B$ implies that $A-B$ is positive definite, i.e. $\lambda_\mathrm{min}(A-B)>0$. The negative definite/semidefinite cases, $\prec$, and $\preceq$, are defined similarly. 
For a Hermitian matrix $A$ and a real function $f\in\mathbb{R}[x]$, the function $f(A)$ of matrix input $A$ is defined by \begin{equation} 
    f(A)=\sum_{\lambda}f(\lambda)\ket{\lambda}\bra{\lambda},
\end{equation} 
where the sum is taken over all the eigenvalues $\{\lambda\}$ of $A$ and $\{\ket{\lambda}\}$ are associated eigenstates. 

We use $\mathcal{D(H)}$ for the space of density matrices on $N$(= $2^n$)-dimensional Hilbert space~$\mathcal{H}$, that is, positive semidefinite matrices with unit trace. Capital $N$ denotes the size of a matrix and lowercase $n$ denotes an integer number of qubits. For simplicity, we only consider matrices of size $N = 2^n$ without loss of generality because a matrix of any size can be trivially extended with zeros such that this is the case.

For a Hermitian operator (Hamiltonian) $H$ acting on $\mathcal{H}$ and a positive real value (inverse temperature) $\beta$, the Gibbs state $\sigma_\beta\in\mathcal{D(H)}$ is defined by
\begin{equation}
\label{eq:Gibbs}
    \sigma_\beta = \frac{e^{-\beta H}}{\tr[e^{-\beta H}]}.
\end{equation}

A thermal pure quantum (TPQ) state associated with the Gibbs state~[Eq.~\eqref{eq:Gibbs}] refers to an $n$-qubit pure state $\ket{\psi_i}$ drawn randomly from an ensemble $\{\ket{\psi_i}\}_i$ such that,
\begin{equation}
\label{eq:TPQ_exp_diff}
    \bra{\psi_i}O_j\ket{\psi_i} - \tr[O_j\sigma_{\beta}]
    \approx 0, 
\end{equation} 
holds with high probability for all $O_j$ from a set of Hermitian operators $\{O_j\}$. 
Importantly, the probability that Eq.~\eqref{eq:TPQ_exp_diff} fails decreases as the system size $n$ grows and vanishes in the asymptotic limit.
See Definition~\ref{def:TPQ_state} for a formal statement of the TPQ state.

 
\section{Semidefinite Programming}
\label{sec:SDP}

 
\subsection{Setup of the Problem}
\label{sec:setup}

We focus on primal semidefinite programs (SDPs) formulated as the maximization of an objective function
\begin{equation}
\label{eq:sdploss}
    L = \tr[CX]
\end{equation}
with respect to the variable $X$, subject to constraints of the form
\begin{equation}
\label{eq:sdpconstr}
\begin{aligned}
    \tr [A_j X]\leq b_j+\varepsilon,&
    \quad
    \forall j = 1,\dots,m
    \\
    X &\succeq 0.
\end{aligned}
\end{equation}
SDPs with $\varepsilon>0$ may be solved with polynomial resources~\cite{lee2015faster} and will be our focus. 
The matrices $C$ and $\{A_j\}_{j=1}^{m}$ have dimension $N\times N$ and are all bounded and Hermitian. In addition, we require the constraint matrices to satisfy $-I\preceq A_j\preceq I$. We label the optimal solution of the SDP as $X_{\mathrm{opt}}$. Other formulations of SDPs, such as where the $\leq$ sign is replaced with an equality sign, or where the maximization is replaced by a minimization, can be mapped to the formulation we consider here.

For the purpose of our quantum algorithm, we define an equivalent SDP by rescaling all the constraints by a constant $R$ such that $\tr[X_\mathrm{opt}/R]\leq 1$, which leads to the constraints of the form, $ \tr [A_jX/R]\leq (b_j+\varepsilon)/R$.  This is necessary since in our quantum algorithm the variable $X$ needs to be embedded in a density matrix, which can be prepared on a quantum computer. 
We then increase the dimension of all the constraint matrices by one
\begin{equation}
\label{eq:extend constraint matrices}
    [A_j]\mapsto
    \begin{bmatrix}
        A_j & 0 \\
        0 & 0 
    \end{bmatrix},
\end{equation}
and introduce a new variable $\omega\geq 0$ such that any solution to the SDP, $X$, can be stored inside a matrix $\rho$,
\begin{equation}
    \rho := 
    \begin{bmatrix}
        X & 0 \\
        0 & \omega 
    \end{bmatrix}.
\end{equation}
The variable $\omega$ is defined such that $\tr[\rho]=1$, and because $X\succeq0$, $\rho$ is a density matrix. The introduction of $\omega$ has no effect on the solution $X_\mathrm{opt}$, or the steps required to solve the SDP.

A consequence of the rescaling is that solving the original SDP with an error tolerance of $\varepsilon$ requires solving the rescaled SDP with an error tolerance of $\epsilon = \varepsilon/R$. This can make the quantum SDP solver unfavourable for problems where $R$ grows fast with problem size.\footnote{Requiring $-I\preceq A_j\preceq I$ and $\tr[X_\mathrm{opt}]\leq 1$ ensures that $-1\leq b_j \leq 1$ for all $j$.}
For the rest of the paper, we assume that the protocol above has been performed and focus on the rescaled SDP where $\tr[X_\mathrm{opt}]\leq 1$
and leave these considerations of specific problems for future work. For notational simplicity, we will leave out the factor of $R$ in the following sections and formulate the complexities in terms of the re-scaled error tolerance $\epsilon$. In the conclusion, sec.~\ref{sec:conclusion}, we will introduce $R$ again in order to compare the complexity of our algorithm to the complexity of other existing quantum SDP algorithms\footnote{A primal SDP also has a dual, which for this primal SDP takes the form
\begin{equation}
\begin{aligned}
    &\min b^Ty\\
    &\text{s.t. }\sum_{j=1}^m y_j A_j-C\succeq 0,\quad y\geq 0
\end{aligned}
\end{equation}
In analogy with the definition of $R$, one can define the constant $r$ as an upperbound on the 1-norm of an optimal solution $y_{\text{OPT}}$, $\lVert y_\text{OPT}\rVert_1\leq r$. An $\frac{\epsilon}{R}$-approximate solution to the rescaled primal SDP is an $\epsilon$-approximate solution to the original primal SDP, but not the dual SDP. To obtain an $\epsilon$-approximate solution to both the primal and dual SDPs, one must find an $\frac{\epsilon}{Rr}$-approximate solution to the rescaled primal SDP~\cite{vanApeldoorn2020quantumsdpsolvers}. For the rest of the paper, we focus on finding $\epsilon$-approximate solutions to the primal SDP only, but this may be changed easily by replacing $R\to Rr$ in Theorem~\ref{thm:total_complexity}.
}.

 
\subsection{Matrix Multiplicative Weights Algorithm and the Zero-Sum Approach}
\label{sec:MMW}

There are many classical and quantum algorithms to solve SDPs~(see e.g.~\cite{Anirudha2020, Augustino2023} and references therein). In our quantum algorithm, we focus on the matrix multiplicative weights (MMW)~\cite{Arora2005} method, whose computational bottleneck is calculating Gibbs-state expectation values. It is this subroutine where quantum computers offer a theoretical speedup. Here, we briefly review the MMW algorithm, and in the next section, we show how to quantise it by using thermal pure quantum states. 

The first step is to reduce the optimization problem to a binary search and feasibility problem. This involves replacing the objective function $L$ [Eq.~\eqref{eq:sdploss}] with a new constraint $\tr[CX]\geq a_{0}-\epsilon$, where $a_0$ is a guess for the optimal value of $L$. One then attempts to find a {\it feasible} matrix $X$, which satisfies all the constraints including the new one. If such an $X$ exists, we increase our guess and repeat because we know the optimal value is greater than $a_{0}-\epsilon$. If no feasible matrix is found, we decrease $a_0$ and repeat. We repeat this process until the desired accuracy is achieved. 
After $\lceil\log(1/\epsilon)\rceil$ iterations of the algorithm combined with a binary search, the final value of $a_0$ is a $\epsilon$-close approximation to the optimal value, $\tr[CX_\mathrm{opt}]$, of the SDP. 

Let $\mathcal{S}_\epsilon$ be the set of all positive semidefinite matrices which satisfy all the constraints, 
\begin{equation}\label{eq:set of satisfying matrices}
    \mathcal{S}_\epsilon 
    = 
    \{ \rho\in \mathcal{D}(\mathcal{H}) 
    |\tr[\rho A_j]\leq b_j+\epsilon
    \text{ for } j = 0,\dots,m\}. 
\end{equation}
Here, $A_0:=-C$ is the new constraint matrix from the objective function and $b_0:=-a_0$. 
The task is to find a feasible $\rho$ by deciding whether $\mathcal{S}_\epsilon$ is empty or not. 
In the MMW method, one iteratively constructs a so-called weight matrix $W_{\tau+1}=\exp(-\gamma\sum_{t=0}^{\tau} M_t)$, where the loss matrices $M_t$ are a particular linear combination of the constraints ${A_i}$, and $\gamma\in\mathbb{R}$ is a constant. At each step $\tau$, one computes a new loss $M_\tau$ by checking the constraints, updates $W_{\tau+1}$ by adding $M_\tau$ to the sum of the previous loss matrices, and constructs a new density matrix $\rho_{\tau+1}=W_{\tau+1}/\tr[W_{\tau+1}]$. If after $T$ iterations we have not found a density matrix that satisfies the constraints, the algorithm decides $\mathcal{S}_\epsilon$ is empty and outputs ``{\it infeasible}''.

There are multiple methods to calculate the loss matrix $M_\tau$, which all broadly have the same computational complexity. 
Here, we focus on the Zero-Sum approach as detailed in~\cite{Lee2015} and given in Algorithm~\ref{alg:zerosum}. In this approach, the loss matrix at step $\tau$ corresponds to the constraint matrix of a broken constraint. Specifically, one sets $\gamma = \frac{\epsilon}{4}$ and $M_\tau =  A_{j^*}$, where $j^*$ is defined as an index such that $\tr[\rho_\tau A_{j^*}]\geq b_{j^*}+\epsilon$. The sum over the loss matrices in the exponent of $W_\tau$ is therefore a parameterized linear combination of the constraints, $W_{\tau+1} = \exp(-\sum_{j = 0}^m \theta_j^{\tau+1} A_j)$, and at each iteration one updates the parameter $\theta_{j^*}^{\tau+1}$: $\theta_{j^*}^{\tau+1} = \theta_{j^*}^{\tau}+\frac{\epsilon}{4}$. One can prove, with approximate Jaynes principle~\cite{Jaynes1957}, that a maximum of $T=\lceil\frac{8}{\epsilon^2}\ln N\rceil$ parameter updates are needed to determine if $\mathcal{S}_\epsilon$ is empty or not. For completeness, we reproduce this proof in Appendix~\ref{app:Jaynes}.  

\begin{algorithm}
\caption{Zero-Sum Algorithm}
\begin{algorithmic}
    \Require $\epsilon>0$, a number of iterations $T = \lceil\frac{8}{\epsilon^2}\ln N\rceil$, $N\times N$ Hermitian matrices $\{A_j\}$ such that $-I\preceq A_j\preceq I$, and real numbers $-1\leq b_j\leq 1$, for $j = 0,\dots,m$.
    \State Let $\rho_0 =\frac{I}{N}$ and $\theta_j^0 = 0$ for all $j\in\{0,\dots,m\}$.
    \For{$\tau = 0,\dots, T$}
    \State Find an $j^*$ such that $\tr[\rho_{\tau}A_{j^*}]>b_{j^*}+\epsilon$
    \If{no $j^*$ exists}
    \State Halt and output ``\textit{feasible}''.
    \EndIf
    \State Set $\theta_{j^*}^{\tau+1}=\theta_{j^*}^{\tau}+\frac{\epsilon}{4}$, and $\theta_{i\neq j^*}^{\tau+1}=\theta_{i\neq j^*}^{\tau}$
    \State Update the weight matrix $W_{\tau+1} = \exp(-\sum_{j=0}^{m} \theta_j^{\tau+1} A_j)$.
    \State Update the density matrix $\rho_{\tau+1} = \frac{W_{\tau+1}}{\tr[W_{\tau+1}]}$
    \EndFor
    \State Output ``\textit{infeasible}''
\end{algorithmic}
\label{alg:zerosum}
\end{algorithm}

As pointed out in~\cite{Brandao2017}, by considering a trivial SDP where finding the optimal value is reduced to a search problem, no generic classical SDP solver including MMW can achieve better complexity than $\tilde{\mathcal{O}}(N+m)$ with constant success probability. For intuition, the factor of $N$ comes from computing the $N\times N$ weight matrix, and the factor of $m$ comes from checking the $m+1$ constraints. In the next section, we show how we can improve on this by replacing the computation of $W_\tau$ with the preparation of TPQ states and checking the constraints with the fast quantum OR lemma~\cite{brandao2019} on a quantum computer. 

 
\section{Quantum MMW Algorithm with TPQ states}
\label{sec:QSDP}

In Ref.~\cite{Brandao2017}, Brand\~{a}o and Svore initiated the development of MMW-based quantum algorithms by observing that $\rho_\tau$ is identified with a Gibbs state [Eq.~\eqref{eq:Gibbs}], whose properties can theoretically be computed faster on a quantum computer~\cite{Yung_2012, Chowdhury_2016, Poulin_2009, zhang2023dissipative, chifang2023quantum}. At the beginning of step $\tau$ in the Zero-Sum algorithm (Algorithm~\ref{alg:zerosum}), 
the weight matrix $W_\tau$ can be written as
\begin{equation}
\label{eq:weight_gibbs}
    W_\tau 
    = 
    \exp\Bigg[-\sum_{j=0}^m\theta_j^\tau A_j\Bigg] 
    = 
    \exp\Bigg[-\frac{\tau \epsilon}{4}\sum_{j=0}^m\frac{4}{\tau \epsilon}\theta_j^\tau A_j\Bigg]
    = 
    \exp\left(-\beta^\tau H^\tau\right),
\end{equation}
where $H^\tau := \sum_{j=0}^m \frac{4}{\tau \epsilon}\theta_j^\tau A_j$ and $\beta^\tau := \tau\epsilon/4$ satisfy $-I\preceq H^\tau\preceq I$ and $\beta^\tau\leq \lceil\frac{2}{\epsilon}\ln N\rceil$, and can be interpreted as Hamiltonian and inverse temperature.
The corresponding Gibbs state $\rho_\tau=W_\tau/\tr[W_\tau]$ is then used to check whether any constraints are broken, i.e., if there exists an index $j^*\in\{0,\dots,m\}$ such that $\tr[\rho_\tau A_{j^*}]>b_{j^*}+\epsilon$.

We now show that we can carry out this task by replacing $\rho_\tau$ with a TPQ state~\cite{Sugiura2012,Coopmans_2023}. 
We first give an algorithm to prepare a TPQ state, which approximates the expectation values $\tr[\rho_\tau A_j]$ up to some problem-dependent error. 
Then, we present a way to find a broken constraint in order to update the weight matrix $W_\tau$ in the Zero-Sum approach. A pseudocode of the full TPQ state quantum SDP algorithm is provided in Algorithm~\ref{alg:TPQ_SDP}.  

\begin{algorithm}
\caption{Quantized MMW Algorithm with TPQ States}
\begin{algorithmic}
    \Require $\epsilon>0$, a number of iterations $T = \lceil\frac{8}{\epsilon^2}\ln N\rceil$, a set of $m+1$ $N\times N$ Hermitian matrices $\{A_j\}$ such that $-I\preceq A_j\preceq I$, and real numbers $-1\leq b_j\leq 1$, for $j = 0,\dots,m$.
    \State Let $\ket{\psi_0} =U\ket{0}$ with a random $k\geq2-$design $U$, and $\theta_j^0 = 0$ for all $j\in\{0,\dots,m\}$.
    \For{$\tau = 0,\dots T$}
    \State Find an $j^*$ such that $\bra{\psi_\tau}A_{j^*}\ket{\psi_\tau}\approx\tr[\rho_{\tau}A_{j^*}]>b_{j^*}+\epsilon$ \hspace{2cm} (Section~\ref{sec:broken_constraint})
    \If{no $j^*$ exists}
    \State Halt and output ``\textit{feasible}''.
    \EndIf
    \State Set $\theta_{j^*}^{\tau+1}=\theta_{j^*}^{\tau}+\frac{\epsilon}{4}$, and $\theta_{i\neq j^*}^{\tau+1}=\theta_{i\neq j^*}^{\tau}$.
    \State Update the SDP Hamiltonian $H^{\tau+1} = \frac{4}{(\tau+1)\epsilon}\sum_{j=0}^{m} \theta_j^{\tau+1}A_j$, and set $\beta^{\tau+1}=\frac{(\tau+1)\epsilon}{4}$.
    \State Update the TPQ state $\ket{\psi_{\tau+1}} = \frac{\exp(-\beta^{\tau+1} H^{\tau+1}/2)U\ket{0}}{\sqrt{\bra{0}U^\dagger \exp(-\beta^{\tau+1} H^{\tau+1})U\ket{0}}}$ \hspace{0.4cm} (Sections~\ref{subsec:TPQ} and~\ref{subsec:prepTPQ})
    \EndFor
    \State Output ``\textit{infeasible}''
\end{algorithmic}
\label{alg:TPQ_SDP}
\end{algorithm}

 
\subsection{Thermal Pure Quantum States}
\label{subsec:TPQ}

TPQ states are pure states that approximate a fixed set of properties of a mixed quantum state of some specific thermodynamic ensemble~\cite{Goldstein2006,Popescu2006,Reimann2007,Sugiura2012,Sugiura2013}. For the canonical ensemble (Gibbs states) we use the following definition. 

\begin{definition}[Thermal pure quantum state~\cite{Sugiura2013}]
\label{def:TPQ_state}
Given a Hamiltonian $H$ acting on $n$ qubits, a constant inverse temperature $\beta>0$ and a set of Hermitian operators $\{O_j\}$. A thermal pure quantum (TPQ) state associated with the Gibbs state $\sigma_\beta=e^{-\beta H}/\tr[e^{-\beta H}]$ is defined by an $n$-qubit pure state $\ket{\psi_i}$ drawn randomly from an ensemble $\{\ket{\psi_i}\}$ such that,
\begin{equation}
\label{eq:TPQ} 
    \mathrm{Pr}\big[\big|\bra{\psi_i}O_j\ket{\psi_i} - \tr[O_j\sigma_{\beta}]\big|
    \geq\xi\big] \leq C_\xi e^{-\alpha n}, 
\end{equation} 
for each $O_j$, and for some constants $C_\xi$ and $\alpha$. 
\end{definition} 

From the definition, it follows that with an exponentially small probability the expectation values of a TPQ state are more than $\xi$ away from the corresponding Gibbs state expectation values. Thus, when the size of the system $n$ increases, the TPQ approximation error vanishes exponentially. This means that for large enough $n$, and the set of operators $\{O_j\}$, we can replace the preparation of a mixed Gibbs state with the preparation of a single randomly sampled pure state (TPQ state). In contrast to a purification of the Gibbs state, which requires $2n$ qubits, a TPQ state is an $n$-qubit state sampled from an ensemble that in the limit $n\to\infty$ matches the canonical Gibbs ensemble for the observables $\{O_j\}$. This property, together with the fact that finite-size pure states require less classical memory (a pure state-vector requires $\mathcal{O}(2^n)$ memory whereas a mixed density matrix requires  $\mathcal{O}(4^n)$ memory), makes TPQ states useful for classically studying physical properties of relatively large Gibbs states~\cite{Sugiura2012,Sugiura2013}. In particular, they have been extensively used in classical numerical simulations of thermodynamic condensed-matter systems, see e.g.~\cite{Jin_2021} and references therein. Here we use them in our quantum SDP solver, which can reduce the required number of qubits compared to preparing purified Gibbs states. 

Inspired by Refs.~\cite{Sugiura2013, Coopmans_2023} for thermodynamic systems, we prove that under some condition on $H^\tau$ (Proposition~\ref{prop:sdp_pur_cond}) the imaginary time evolved states, 
\begin{equation}
\label{eq:psi_beta}
    \ket{\psi_\tau} = \frac{e^{-\beta^\tau H^\tau/2}U\ket{0}}{\sqrt{\bra{0}U^\dagger e^{-\beta^\tau H^\tau}U\ket{0}}},
\end{equation} 
satisfies Def.~\ref{def:TPQ_state}, hence can be used to approximate the Gibbs state expectation values in our quantum SDP solver. Here $U$ is randomly drawn from a unitary $k$-design for an integer $k\ge2$. For example, $U$ can be uniformly drawn from the $n$-qubit Clifford group, which can be done classically efficiently~\cite{PhysRevA.70.052328}.

We sketch the idea of the proof here and delegate the detailed discussion to Appendix~\ref{app:tpq_deriv}. 
First, by averaging over the unitary $k$-design, we can bound the mean-squared error in the expectation value of the constraint matrices $A_j$,  
\begin{equation}
\label{eq:TPQ_mse_main}
\begin{aligned}
\mathbb{E}_U[(\bra{\psi_\tau}A_j\ket{\psi_\tau}-\tr[\rho_\tau A_j])^2]
    \le
    \frac{105}{2}(\tr[\rho_{\tau}^2])^{1/2}
    + \calO[(\tr[\rho_{\tau}^2])^{2/3}].
\end{aligned}
\end{equation} 
To the best of our knowledge, this is the first rigorous bound on the error in canonical TPQ-state expectation values. 
From a Markov inequality, we obtain 
\begin{align}
\label{eq:sdptpq}
\begin{split}
\mathrm{Pr}\big[|\bra{\psi_\tau}A_j\ket{\psi_\tau} 
    - \tr[A_j\rho_\tau]|\ge \xi\big]
&\leq\frac{\mathbb{E}_U[(\bra{\psi_\tau}A_j\ket{\psi_\tau}-\tr[\rho_\tau A_j])^2]}{\xi^2} \\
    & \leq \frac{\frac{105}{2}(\tr[\rho_{\tau}^2])^{1/2}
    + \calO\left[(\tr[\rho_{\tau}^2])^{2/3})\right]}{\xi^2}.
\end{split}
\end{align}

Comparing  Eq.~\eqref{eq:TPQ} to Eq.~\eqref{eq:sdptpq}, we find that $\ket{\psi_\tau}$ satisfies the definition with $C_\xi=\calO(1/\xi^2)$ if $\tr[\rho_{\tau}^2]=\calO(e^{-\alpha n})$.
In other words, $\ket{\psi_\tau}$ is a TPQ state when the purity $\tr[\rho_{\tau}^2]$ of the Gibbs state vanishes exponentially with $n$. For physical Hamiltonians $H^\tau$, i.e., those for which the Gibbs state has an extensive free energy, we can adopt the arguments from~\cite{Sugiura2013} to show that the purity vanishes with system size. For completeness, we reproduce this argument in Appendix~\ref{sec:free-energy}. For generic SDPs we can, however, not guarantee that $H^\tau$ satisfies the property. For this scenario, we propose the following sufficient condition on the spectrum of $H^\tau$.

\begin{proposition}[a spectral condition for vanishing purity]
\label{prop:sdp_pur_cond}
Let $c\in[0,1]$ and $\nu\in(0,\epsilon\frac{\ln 2}{4})$ with $\epsilon>0$.
Given a Hermitian matrix $H^\tau$ of size $2^n\times2^n$ that has $c2^n$ eigenvalues in the range $[\lambda^\tau_{\mathrm{min}}, \lambda^\tau_{\mathrm{min}} + \nu ]$, the purity of the Gibbs state $\rho_\tau = e^{-\beta^\tau H^\tau}/\tr[e^{-\beta^\tau H^\tau}]$ with $\beta^\tau=\tau\epsilon/4$ can be upper bounded by $\tr[\rho_{\tau}^2]\leq\frac{2^{-\left(1-\frac{4\nu}{\epsilon\ln2}\right)n}}{c^2}$ for all $0\le\tau\leq \frac{8n}{\epsilon^2}$.
\end{proposition}

We provide the proof together with a more generic form of the condition in Appendix~\ref{app:exponential decaying purity}. The condition shows that if a constant fraction $c$ of the eigenvalues of the rescaled Hamiltonian $H^\tau$ are within a constant $\nu$ from the lowest eigenvalue, then the purity vanishes exponentially with system size. In other words, the condition is met when there is a high density of low-energy states. 
While this might appear abstract, we conjecture it to be satisfied for a large class of matrices $H^\tau$. This is because the gaps between the eigenvalues of any matrix $H^\tau$ are on average  exponentially small in $n$ (it has $2^n$ eigenvalues in the interval $[-1,1]$). In order to support our claim, in Appendix~\ref{sec:app_cond_hams} we show that with probability $1-\mathcal{O}(e^{-n})$ a random matrix drawn from the generalized unitary ensemble (GUE)~\cite{erdos2017} satisfies the condition. These matrices share some properties with chaotic quantum systems and random Pauli-string Hamiltonians~\cite{Chen2023}. We also verify the spectral property for the Hamiltonians we use in the Hamiltonian learning experiments in Section~\ref{sec:numerics}.\footnote{
In this case, we can also resort to the free-energy argument for physical Hamiltonians.
}

With the condition, or for a physical SDP Hamiltonian, we can use the expectation values of the TPQ states $\ket{\psi_\tau}$ as approximators of the Gibbs-state expectation values in the MMW algorithm. 
Importantly, the success probability in Eq.~\eqref{eq:sdptpq} depends on the specific form of $H^\tau$ (the density of low energies in its spectrum), the target SDP error $\epsilon\sim\xi$, and the norm and size $N$ of the constraint matrices (which we take to be at most 1 in our problem formulation~\ref{sec:SDP}). 
This means that a priori we cannot deterministically bound the error in all cases.
However, for problems whose Hamiltonians satisfy the condition in Proposition \ref{prop:sdp_pur_cond}, we know that the purity, and hence the success probability decreases with system size as $N^{-\alpha}$ for some constant $\alpha>0$. This is particularly relevant for the Hamiltonian learning problem in the regime of problem sizes that are too hard to solve classically.

 
\subsection{Preparing Thermal Pure Quantum States on a Quantum Device}
\label{subsec:prepTPQ}

Having obtained a set of TPQ states that can be used in our quantum SDP solver, we now present a quantum algorithm to prepare them. To this end, we employ quantum eigenvalue transformation (QET)~\cite{Martyn2021grand, Gilyen2019}. QET is a framework for performing matrix arithmetic on quantum computers\footnote{We remark on closely related notions, quantum signal processing (QSP) and the quantum singular value transformation (QSVT). QSP refers to the single-qubit version of QET and QSVT is the generalization of QET to non-diagonalizable matrices.}, which recently has been demonstrated in small-scale quantum hardware experiments~\cite{Dong_2022, Kikuchi2023, debry2023experimental}. QET allows us to construct a quantum circuit for a polynomial approximation to the matrix exponential $e^{-\beta^\tau H^\tau/2}$ given access to a block-encoding circuit for the Hermitian matrix $H^\tau$. It can be combined with a random circuit $U$, such as a random Clifford circuit, to obtain an end-to-end circuit implementation of $\ket{\psi_\tau}$. 
We first introduce a block encoding of a bounded matrix.  
\begin{definition}[Block encoding of a square matrix]
\label{def:block_encoding} 
    For an $2^n\times2^n$ square matrix $A$ such that $\lVert A\rVert\le1$, a $(n+a)$-qubit unitary $U_A$ is a block encoding of $A$ if
    \begin{equation}
    \label{eq:block encoding}
        (I\otimes\bra{0^a})U_A(I\otimes\ket{0^a}) = A.
    \end{equation}
\end{definition} 
In other words, a block-encoding is an embedding of an arbitrary bounded matrix $A$ into a subblock of a larger unitary matrix $U_A$.

In our quantum algorithm, we make use of the block-encoding circuit $U_{K^\tau}$ of the shifted Hamiltonian 
\begin{equation}
\label{eq:block_encoding_Ktau}
    K^\tau := H^\tau-(1+\Xi^\tau)I,
\end{equation}
where $\Xi^\tau$ is an approximation of the lowest eigenvalue of $H^\tau$ satisfying $\lambda_\mathrm{min}^\tau-1/2\beta^T\leq\Xi^\tau\leq\lambda_\mathrm{min}^\tau$ when $\lambda_\mathrm{min}\geq -1+1/2\beta^T$ and $\Xi^\tau=0$ otherwise.\footnote{
We shift the smallest eigenvalue of the Hamiltonian to $\lambda_{\mathrm{min}}\geq -1 + 1/2\beta^T$ in order to lower bound the success probability of our algorithm. This is explained in Appendix~\ref{app:prob_lower_bound}.
}
Refs.~\cite{Gilyen2019,Low2019hamiltonian,Chakraborty2019power,Camps:2022jnx,Zhang2022,Sunderhauf:2023xrz} provide block-encoding methods for a wide range of problems applicable to our setting. For any matrix satisfying the conditions in Definition \ref{def:block_encoding}, a block encoding always exists, but as in other block-encoding-based algorithms, to realize the quantum speedup the block encoding of $K^\tau$ must be efficient.\footnote{By efficient we mean that the circuit depth is at most $\mathcal{O}[$polylog$(N)]$ and requires at most $\mathcal{O}[$polylog$(N)]$ ancillary qubits.} See \cite{Dalzell2023} for a comprehensive survey.

The block-encoding $U_{K^\tau}$ can then be used to construct a circuit for $e^{-\beta^\tau H^\tau}$ with QET.

\begin{lemma}[Quantum eigenvalue transformation of indefinite parity, restatement of Theorem 56 in Ref.~\cite{Gilyen2019}]\label{lem:QET}
    Let $U_A$ be a $(n+a)$-qubit block encoding of a Hermitian matrix $A$. Let $P\in \mathbb{R}[x]$ be a real polynomial of degree $d$ such that $\max_{x\in[-1,1]}|P(x)|\le\frac{1}{2}$. Then there exists a block encoding $U^{\vec{\phi}}_\mathrm{QET}$ of $P(A)$ that uses $2d$ queries to $U_A$ and $U_A^\dag$, a single application of controlled-$U_A$, and $\mathcal{O}((a+1)d)$ other elementary gates. Moreover, a set of $2d+1$ angles $\vec{\phi}=\{\phi_0,\dots,\phi_{2d}\}$ parameterising $U^{\vec{\phi}}_\mathrm{QET}$ can be computed classically efficiently.
\end{lemma}

Therefore, we need to find a suitable polynomial $P^\mathrm{exp}_\beta(x)$ which approximates $e^{-\beta x}$ for $x\in[-1,1]$. 
We start with the following Lemma for a polynomial approximation of an exponential function. 

\begin{lemma}[Polynomial approximation of exponential~\cite{Sachdeva2014}, Lemma~4.2]
\label{lemma:exp_poly_approx}
Let $\beta\in\mathbb{R}_+$ and $\mu\in(0,\frac{1}{2}]$. There exists a polynomial $P^\mathrm{exp}_\beta\in\mathbb{R}[x]$ satisfying
\begin{equation}
    \max_{x\in[-1,1]}|P^\mathrm{exp}_\beta(x) - e^{-\beta(x+1)}|
    \leq\mu,
\end{equation}
and the degree of $P^\mathrm{exp}_\beta$ is $\mathcal{O}\big(\sqrt{\beta}\log\frac{1}{\mu}\big)$. 
\end{lemma}

Combining this with the Lemma~\ref{lem:QET}, we arrive at the following lemma.
\begin{lemma}[Approximating the matrix exponential with QET]
\label{lem:qet_matrixexp}
    Given a block encoding $U_{M}$ of a Hermitian matrix $M$ such that $\lVert M\rVert\le1$, one can construct a block encoding $U_{P^\mathrm{exp}_{\beta}(M)/2}$ of $P^\mathrm{exp}_{\beta}(M)/2$ such that
    \begin{equation}
    \label{eq:QET_error}
        \big\lVert P^\mathrm{exp}_\beta(M)
        - e^{-\beta (M+I)} \big\rVert
        \leq 
        \mu
    \end{equation}
    with $\calO(\sqrt{\beta}\log\frac{1}{\mu})$ uses of $U_{M}$ and $U_{M}^\dag$ and other elementary gates.
\end{lemma} 

Setting $M= K^\tau$ and $\beta=\beta^\tau/2$ in Lemma~\ref{lem:qet_matrixexp} gives a block-encoding circuit $U_{P^\mathrm{exp}_{\beta^\tau/2}(K^\tau)/2}$ for the matrix exponential  $e^{-\frac{\beta^\tau}{2} (K^\tau +I)}/2$. Successful implementation of $U_{P^\mathrm{exp}_{\beta^\tau/2}(K^\tau)/2}$ requires measuring the ancillary register in the all-zero state, which has probability
\begin{equation}
\label{eq:prob_exp}
    p_\text{exp} 
    \approx
    \frac{1}{4}\bra{0}U^\dagger e^{-\beta^\tau (K^\tau+I)} U\ket{0}
    =
    \frac{1}{4}\bra{0}U^\dagger e^{-\beta^\tau(H^\tau-\Xi^\tau I)} U\ket{0}.
\end{equation}
The approximate equality is due to the QET error in Eq.~\eqref{eq:QET_error}, which can be bounded as we elaborate in Appendix~\ref{app:tpq_complexity}.
This probability can be boosted to $\Omega(1)$ by performing $\calO(p_\text{exp}^{-1/2})$ rounds of amplitude amplification~\cite{Martyn2021grand}. The gate complexity of the circuit is therefore $\calO(p_\text{exp}^{-1/2}) \times \mathcal{O}(\sqrt{\beta^\tau}\log{\frac{1}{\mu}})\times \mathcal{T}_{K^\tau}$, where $\mathcal{T}_{K^\tau}$ is the gate complexity of the block encoding of $K_\tau$ [Eq.~\eqref{eq:block_encoding_Ktau}]. By bounding $p_{\mathrm{exp}}$ and the approximation errors, as shown in Appendix~\ref{app:tpq_complexity}, we arrive at the following theorem for the preparation of $\ket{\psi_\tau}$. 

\begin{theorem}[Circuit implementation of a TPQ state]
\label{thm:tpq_complexity}
    Let $\tau\in[0, T]$ with $T\le \lceil\frac{8\ln N}{\epsilon^2}\rceil$, and $U$ a random unitary drawn from a unitary $k\ (\ge2)$-design. Also, let $\xi\in[0,1]$ and assume that the maximum purity over all $\tau$ of the Gibbs state $\rho_\tau=e^{-\beta^\tau H^\tau}/\tr[e^{-\beta^\tau H^\tau}]$ decays exponentially as $\tr[\rho_\tau^2]\propto e^{-\alpha \log{N}}$ for some positive constant $\alpha$ as the size of the program $N$ increases.  Then, there exists a quantum operation $V_{\beta^\tau}(H^\tau)$ to create the state
    \begin{equation}
    \label{eq:approx_TPQ}
        \ket{\tilde{\psi}_\tau} 
        =  V_{\beta^\tau}(H^\tau)U\ket{0},
    \end{equation}
    which acts as a TPQ state, i.e., for a predefined set of operators $\{T_j\}$ with $\lVert T_j\rVert \leq 1$,
    \begin{equation}
    \label{eq:TPQ_with_QET}
        \mathrm{Pr}_{U}
        \big[|\bra{\tilde{\psi}_\tau}T_j\ket{\tilde{\psi}_\tau}
        -\tr[\rho_\tau T_j]|\geq \xi\big]
        \leq C_\xi e^{-\alpha n},
    \end{equation}
    for some $\xi$-dependent constant $C_\xi$.
    This quantum operation can be implemented using at most
    \begin{equation}
    \label{eq:TPQ_complexity}
        \mathcal{O}\left(\frac{(N\log{N})^{1/2}}{\xi^{1/2}}\log{\frac{N}{\xi}} \right)
    \end{equation}
     elementary gates and block encodings $U_{K^\tau}$ and $U_{K^\tau}^\dag$ of the shifted Hamilontian operator defined in Eq.~\ref{eq:block_encoding_Ktau}, and $\mathcal{O}(\log N)$ qubits.
\end{theorem} 

The complexity of estimating the Gibbs-state expectation values in our quantized MMW method therefore scales with $\tilde{\mathcal{O}}(\sqrt{N})$. This is a quadratic improvement over the classical lower bound, which holds for generic cases (without vanishing purity). Compared to the best previous quantum SDP solver~\cite{vanapeldoorn2019}, our TPQ state approach reduces the number of ancillary qubits and has an improved dependence on the SDP error with a factor $\epsilon^{-1/2}$. However, we remark that the improved  $\epsilon$-dependence stems from our assumption on an available block-encoding circuit $U_{K^\tau}$. 
In practice, finding such a block encoding can be non-trivial. In Ref.~\cite{vanapeldoorn2019,vanApeldoorn2020quantumsdpsolvers,Gilyen2019} the authors claim that using a different block encoding and polynomial approximation one could get an implementation of $e^{-\beta^\tau H^\tau/2}$ using $\sim\epsilon^{-1}$ queries to the block encoding. In appendix~\ref{app:alternative_QETexp} we provide an explicit LCU-based implementation which uses $\sim\epsilon^{-3/2}$ queries.

We can further improve the complexity in $N$ when the spectral condition in Proposition~\ref{prop:sdp_pur_cond} is satisfied. 

\begin{corollary}
\label{cor:tpqcomplex}
If we additionally assume that the Hamiltonian $H^\tau$ has $c2^{n}$ eigenstates within the interval $[\lambda_\mathrm{min}(H^\tau),\lambda_\mathrm{min}(H^\tau)+n\nu]$ with some constants $c\in[0,1]$ and $\nu\in(0,\epsilon \ln{2}/4)$. Then the quantum operation $V_{\beta^\tau}(H^\tau)$  can be implemented using at most \begin{equation}
    \label{eq:improved_TPQ_complexity}
        \mathcal{O}\left(\frac{(N\log^2{N})^{1/4}}{\epsilon^{1/2}}\log{\frac{N}{\xi}} \right)
    \end{equation}
     elementary gates and block encodings $U_{K^\tau}$ and $U_{K^\tau}^\dag$ of the shifted Hamilontian operator defined in Eq.~\ref{eq:block_encoding_Ktau}, and $\mathcal{O}(\log N)$ qubits.
\end{corollary}

The spectral condition (Proposition~\ref{prop:sdp_pur_cond}) leads to the vanishing purity that is assumed in Theorem~\ref{thm:tpq_complexity} as discussed in Appendix~\ref{app:exponential decaying purity}.
In this case, the number of queries to the block encoding scales with $N^{1/4}$, which results from our lower bound on the success probability given by Eq.~\eqref{eq:prob_exp} for Hamiltonians satisfying the spectral condition. This may imply a quartic quantum speedup over the classical complexity lower bound of $\tilde{\mathcal{O}}(N)$ for computing Gibbs state expectation values. However, there are two open questions. Firstly, in order to make this rigorous one needs to prove a classical lower bound in the presence of the spectral condition. We leave this question open for future work, but highlight that naive classical algorithms which make use of sparse matrix-vector multiplication have complexity $\tilde{\mathcal{O}}(N)$. Secondly, our algorithm requires an estimate of the ground state energy. There are numerous algorithms for this~\cite{Poulin2009,Ge2019,vanApeldoorn2020quantumsdpsolvers,Lin2020,Lin2022}, which have a generic worst case complexity of $\tilde{\mathcal{O}}(\sqrt{N})$. It is unclear if this can be improved for cases where the spectral condition is satisfied. Hence, the overall complexity still scales with $\sqrt{N}$, which is the same as the quantum lower bound $\tilde{\mathcal{O}}(\sqrt{N})$ for the generic case. Nevertheless, the ground state estimation only needs to be done once per iteration of the MMW algorithm, whereas the Gibbs state expectation value computation many times, which may make our algorithm more efficient in practice despite the same asymptotic scalings.

 
\subsection{Broken constraint check}
\label{sec:broken_constraint}

The next step in the algorithm is finding a broken constraint if one exists, or correctly determining that all constraints are satisfied (recall Algorithm~\ref{alg:TPQ_SDP}). One could measure all $m+1$ expectation values and check the constraints one by one, which would require at least $\calO(m)$ preparations of the TPQ state and measurements. In this section, we show we can achieve a quadratic speedup with respect to $m$. This requires combining the fast quantum OR lemma~\cite{brandao2019} with our construction for the TPQ states described in Section~\ref{subsec:TPQ}. 
We remark that this part of the algorithm is independent of the specific estimator of Gibbs-state expectation values. For instance, it can also be combined with purified Gibbs states. In fact, our algorithm is inspired by the constraint check in~\cite{vanApeldoorn2020quantumsdpsolvers}. We provide an explicit implementation and combine it with the TPQ states from the previous section. 

We start by introducing the quantum OR lemma and giving a high-level outline of our proof and techniques.  
\begin{lemma}[Fast quantum OR lemma~\cite{brandao2019,Gilyen2019}]
\label{lemma:fastOR_QET}
Let $\eta$ be an unknown density matrix, $\Pi_1,\dots,\Pi_k$ be a sequence of projection operators and $\upsilon,\zeta\in[0,1/2)$. 
Suppose we are promised that either
\begin{enumerate}
	\item[(i)] there exists a $j\in\{1,\dots,k\}$ such that $\tr[\Pi_j\eta]\ge 1-\upsilon$, or
	\item[(ii)] $\frac{1}{k}\sum_{j=1}^{k}\tr[\Pi_j\eta]\le \phi$.
\end{enumerate}
Then, there is a test that uses one copy of $\eta$, and that accepts with probability at least $(1-\upsilon)^2/4-\zeta$ in case $(i)$ and with probability at most $5k\phi+\zeta$ in case $(ii)$.
The test requires $\mathcal{O}(\log(1/\zeta)\sqrt{k})$ uses of the operator $\sum_{j=1}^{k}\mathrm{C}_{\Pi_j}\mathrm{NOT} \otimes \ket{j}\bra{j}$, $\mathcal{O}(\log(1/\zeta)\sqrt{k}\log k)$ other gates, and $\calO(\log k)$ ancillary qubits.
\end{lemma}

The OR lemma can be used to distinguish between two scenarios depending on the input state $\eta$ and the projection operators $\Pi_j$. In the first scenario, the expectation value of one of the projectors is greater than some fixed value. In the second scenario, the average expectation value of all projectors is below some other fixed value. In both cases, the test outputs ``1'' (acceptance) with a certain probability. By choosing the right parameters $\nu, \zeta$ the probability of acceptance in scenario $(i)$ is larger than the probability in scenario $(ii)$. Hence, by repeating the test enough times, one can distinguish between the two cases with high probability. 

For our SDP algorithm, we wish to distinguish between the case when there is a broken constraint, i.e., there is an index $j$ such that $\tr[\rho_\tau A_j]-b_j\geq \epsilon$, and the case $\tr[\rho_\tau A_j]-b_j\leq \epsilon$ for all $j\in\{0,\dots,m\}$. In the following, we construct $k=m+1$ projectors $\Pi_j$ which combined with the OR lemma can be used to achieve this task. We first do this for the Gibbs state $\rho_\tau$, then substitute our TPQ states and derive how many times we need to repeat the OR lemma test. This leads to the total circuit complexity of our TPQ-SDP solver in terms of queries to the block encoding of $K^\tau$ [Eq.~\eqref{eq:block_encoding_Ktau}] and other elementary gates. An abstract circuit diagram of the broken constraint check is given in Fig.~\ref{fig_projector}. For the interested readers, in Appendix~\ref{app:orlemma} we also include the test and proof of Lemma~\ref{lemma:fastOR_QET}. 

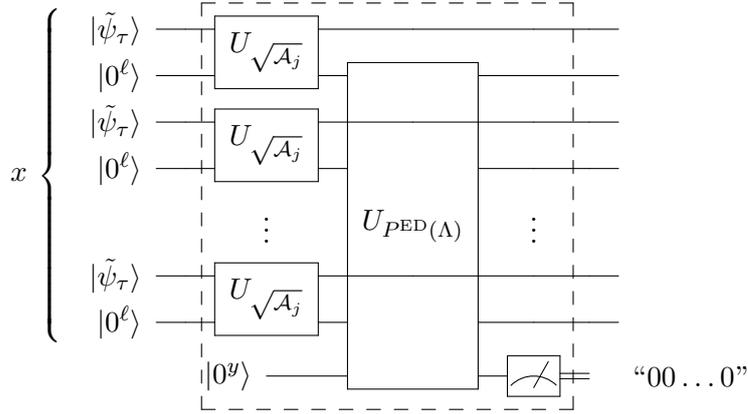
\begin{figure}
\begin{align}
\nonumber
\label{eq:ctrlU}
    \begin{array}{c}
        \Qcircuit @C=1.em @R=0.7em {
        &&&\lstick{\ket{\tilde{\psi}_{\tau}}}
        & \qw
        & \multigate{1}{U_{\sqrt{\calA_j}}}
        & \qw 
        & \qw
        & \qw
        & \qw
        \\
        &&&\lstick{\ket{0^\ell}}
        & \qw
        & \ghost{U_{\sqrt{\calA_j}}}
        & \multigate{8}{U_{P^{\text{ED}}(\Lambda)}}
        & \qw
        & \qw
        & \qw
        \\
        &&&\lstick{\ket{\tilde{\psi}_{\tau}}}
        & \qw
        & \multigate{1}{U_{\sqrt{\calA_j}}}
        & \qw 
        & \qw
        & \qw
        & \qw
        \\
        &&&\lstick{\ket{0^\ell}}
        & \qw
        & \ghost{U_{\sqrt{\calA_j}}}
        & \ghost{U_{P^{\text{ED}}(\Lambda)}}
        & \qw
        & \qw
        & \qw
        \\
        \\
        &&&
        &
        & \vdots
        &
        & \vdots
        \\
        \\
        &&&\lstick{\ket{\tilde{\psi}_{\tau}}}
        & \qw
        & \multigate{1}{U_{\sqrt{\calA_j}}}
        & \qw
        & \qw
        & \qw
        & \qw
        \\
        &&&\lstick{\ket{0^\ell}}
        & \qw
        & \ghost{U_{\sqrt{\calA_j}}}
        & \ghost{U_{P^{\text{ED}}(\Lambda)}}
        & \qw
        & \qw
        & \qw
        \\
        &&&
        &
        & \lstick{\ket{0^y}}
        & \ghost{U_{P^{\text{ED}}(\Lambda)}}
        & \meter{}
        & \cw
        & \rstick{``00\dots0"}
        \gategroup{1}{6}{10}{8}{.9em}{--}
        \inputgroupv{1}{9}{1.3em}{5.em}{x}
        }
    \end{array}
\end{align}
\caption{\label{fig_projector} 
The circuit diagram of the broken constraint check. The input state is $\tilde{\eta}$~[Eq.~\eqref{eq:eta_tilde}] is a tensor product of $x$ copies of the TPQ state $\ket{\tilde{\psi}_{\tau}}$ and $y$ additional ancillary qubits. Acting on this state with the circuit surrounded by the dashed box and post-selecting with ``$00\dots0$'' in the bottom register correspond to applying the projector $\Pi_j$ [Eq.~\eqref{eq:amp_proj}]. This consists of first applying the block encoded constraint matrices $U_{\sqrt{\calA_j}}$ followed by eigenvalue discrimination $U_{P^{\text{ED}}(\Lambda)}$ of the majority-voting operator~$\Lambda$. The post-selection succeeds with probability $\tr[\tilde{\eta}\Pi_j]\approx 1$ when a constraint $j$ is broken [case $(i)$] and $\tr[\tilde{\eta}\Pi_j]\approx 0$ when constraint $j$ is not broken [case $(ii)$].}
\end{figure}

 
\subsubsection{Constructing the OR Lemma Projectors}

In order to construct the projectors $\Pi_j$ in Lemma~\ref{lemma:fastOR_QET} that correspond to the constraint matrices $A_j$ of the SDP, we first shift and rescale them
\begin{equation}
\label{eq:calA}
    \calA_j:=\frac{(2-b_j)I+A_j}{4},
\end{equation}
such that $0\preceq \calA_j\preceq I$ for each $j\in\{0,\dots,m\}$ and the constraints can be written as $\tr[\rho_\tau \mathcal{A}_j]\leq 1/2+\epsilon/4$.
This allows us to construct a block encoding $U_{\sqrt{\calA_j}}$ of the square root of the constraint matrix $\sqrt{\calA_j}$ using $\ell$ ancillary qubits, i.e., $(I\otimes\bra{0^\ell})U_{\sqrt{\calA_j}}(I\otimes\ket{0^\ell}) = \sqrt{\calA_j}$~\cite{Somma2013,Chowdhury_2016}. 

Applying $U_{\sqrt{\calA_j}}$ to the state $\rho_\tau\otimes \ket{0^\ell}\bra{0^\ell}$ and measuring the ancillary qubits in $\ket{0^\ell}\bra{0^\ell}$, we find
\begin{equation}
    \tr\Big[
    U_{\sqrt{\calA_j}}
    \big(\rho_\tau \otimes \ket{0^\ell}\bra{0^\ell}\big)
    U_{\sqrt{\calA_j}}^\dagger 
    \big(I^{\otimes n}\otimes \ket{0^\ell}\bra{0^\ell}\big) 
    \Big] 
    = 
    \tr\left[ \rho_\tau \mathcal{A}_j \right]. 
\end{equation} 
The probability that the ancillary register is found in state $\ket{0^\ell}$ is equal to expectation value $\tr\left[\rho_\tau \mathcal{A}_j\right]$. Making use of the Chernoff bound,\footnote{
Chernoff bounds (one-sided Chernoff-Hoeffding inequalities) are given by
\begin{equation}
    \mathrm{Pr}\left(X-\mu \geq \delta\right)\le e^{-2\delta^2/n},
    \qquad
    \mathrm{Pr}\left(X-\mu \geq -\delta\right)\ge 1-e^{-2\delta^2/n}, 
\end{equation}
where $X$ is the sum of $n$ independent random variables taking values in $[0, 1]$, $\mu$ is the sum's expectation value, and $\delta>0$.
}
we can show that, by measuring the ancillary register $x=\calO(\frac{1}{\epsilon_\text{gap}^2}\ln\frac{m}{\delta})$ times, we can determine if $\tr\left[\rho_\tau \mathcal{A}_j\right]\leq 1/2 +\epsilon/4$ up to an error $\epsilon_\text{gap}$ with failure probability at most $\delta/m$. 

To turn this into a quantum algorithm, we construct a single projector for each constraint matrix that acts on $x$ copies of the state $\rho_\tau\otimes \ket{0^\ell}\bra{0^\ell}$,
\begin{equation}
\label{eq:eta_tilde}
    \tilde{\eta} 
    = \big(\rho_\tau \otimes \ket{0^\ell}\bra{0^\ell}\big)^{\otimes x}.  
\end{equation} 
We then introduce the operator
\begin{equation}
\label{eq:amp_projector1}
    \Lambda=\frac{1}{x}\sum_{i=0}^{x-1}\Lambda_i,
\end{equation} 
where $\Lambda_i=\ket{0^\ell}\bra{0^\ell}$ acts on the ancillary register of $i^{\text{th}}$ copy. The operator $\Lambda$ computes the number of times the ancillary registers are projected onto $\ket{0^\ell}$ divided by the number of copies~$x$. By applying the block-encoding $U_{\sqrt{A_j}}$ to each copy of $\tilde{\eta}$ and measuring $\Lambda$, we obtain the expectation value 
\begin{equation}
    \tr\Big[
    \big(U_{\sqrt{A_j}}\big)^{\otimes x}\, \tilde{\eta}\,
    \big(U_{\sqrt{A_j}}^\dagger\big)^{\otimes x} \Lambda
    \Big]
    =
    \tr\left[\rho_\tau \mathcal{A}_j\right].
\end{equation}

We can determine if this expectation value is greater than or below the threshold, $1/2+\epsilon/4$, up to an error $\epsilon_\text{gap}$ by performing eigenvalue projection on the input state $\tilde{\eta}$. This is stated in the following Lemma, which is proven with another QET routine in Appendix~\ref{app:eigendisclemma}. 

\begin{lemma}[Eigenvalue projection]
\label{lemma:eigenvalue_disc_density}
    Let $b>a>0$, $\delta'\in(0,1)$, $p_a,p_b \in[0,1]$, and $U_\Lambda$ be a block encoding of a Hermitian matrix $\Lambda$. Suppose a quantum state~$\rho$ is promised to satisfy either $\tr[\rho \Pi_{\lambda\ge b}]\ge p_b$ or $\tr[\rho \Pi_{\lambda\le a}]\ge p_a$ where $\Pi_{\lambda\geq b}$ is a projector onto the eigenspace of $\Lambda$ with eigenvalues $\geq b$, and $\Pi_{\lambda\leq a}$ is defined similarly.
    Then, there is an algorithm that accepts $\rho$ with probability $p_b (1-\delta')^2$ in case $\tr[\rho \Pi_{\lambda\ge a}]\ge p_b$ holds, and accepts it with probability $\delta'^2 + (1-p_a)$ in case $\tr[\rho \Pi_{\lambda\le a}]\ge p_a$ holds. The algorithm is implemented with $\mathcal{O}(\frac{1}{b-a}\log\frac{1}{\delta'})$ uses of $U_\Lambda$ and other elementary gates.
\end{lemma}

We implement this Lemma by constructing an approximate projector, $P^{\mathrm{ED}}(\Lambda)$,\footnote{
The projector is constructed with the approximate step function $P^\text{ED}(x)$ that satisfies $|P^\text{ED}(x)|\le1$ for all $x\in[-1,1]$ and
\begin{equation}
\left\{
\begin{array}{ll}
    P^\text{ED}(x)\ge 1- \delta'
    & x\in[-1,-b]\cup[b,1],
    \\[.2em]
    |P^\text{ED}(x)|\le \delta'
    & x\in[-a,a].
\end{array}
\right.
\end{equation}
The superscript ``ED'' stands for eigenvalue discrimination.
} 
on the eigenvalues of the operator $\Lambda$. 
Setting $a=1/2+\epsilon/4$, $b=1/2+\epsilon/4+\epsilon_{\text{gap}}$, $\delta'=\delta/(m+1)$, and using a block-encoding $U_\Lambda$ of $\Lambda$ (Appendix~\ref{app:majorblock}), we can construct a circuit $U_{P^{\mathrm{ED}}(\Lambda)}$ to distinguish between $\tr\left[\rho_\tau \mathcal{A}_j\right]\geq 1/2+\epsilon/4+\epsilon_\text{gap}$ and $\tr\left[\rho_\tau \mathcal{A}_j\right]< 1/2 +\epsilon/4$ with failure probability $\calO(\delta/m)$. The projectors
\begin{equation}
\label{eq:amp_proj}
    \Pi_j 
    =
    \big(U_{\sqrt{\calA_j}^\dag}\,^{\otimes x}\otimes I^{y}\big)
    U_{P^{\text{ED}}(\Lambda)}^\dag
    (I\otimes\ket{0^y}\bra{0^y})
    U_{P^{\text{ED}}(\Lambda)}
    \big(U_{\sqrt{\calA_j}}\,^{\otimes x}\otimes I^{y}\big),
\end{equation}
are then suitable for use in the quantum OR Lemma (Lemma~\ref{lemma:fastOR_QET}). Here, the state $\ket{0^y}$ is the $y=\lceil \log{x}\rceil +1$-qubit ancillary register for the block-encoding of $\Lambda$. 
Case $(i)$ in the OR Lemma corresponds to $\tr\left[\rho_\tau \mathcal{A}_j\right]\geq 1/2+\epsilon/4 +\epsilon_\text{gap}$, and case $(ii)$ to $\tr\left[\rho_\tau \mathcal{A}_j\right]< 1/2 +\epsilon/4$. 

The projector $\Pi_j$ can be implemented with the gate complexity,
\begin{equation}
\label{eq:Vi_complexity}
    \mathcal{T}_{\Pi_j} 
    = 
    \mathcal{O}\left(
    \frac{\log (m/\delta)}{\epsilon_\text{gap}^2}
    \left(\mathcal{T}_{\sqrt{\calA_j}}
    +\frac{\log (m/\delta)}{\epsilon_\text{gap}}\right)
    \right),
\end{equation} 
where $\mathcal{T}_{\sqrt{\calA_j}}$ is the gate complexity of block encoding the matrix $\sqrt{\calA_j}$. 
Hence, the operator $\sum_{j=0}^{m}\mathrm{C}_{\Pi_j}\mathrm{NOT} \otimes \ket{j}\bra{j}$ required for Lemma~\ref{lemma:fastOR_QET} is constructed with $\calO\big(\frac{1}{\epsilon_\text{gap}^2}\log \frac{m}{\delta}\big)$ uses of the operator $\sum_{j=0}^{m}{U_{\sqrt{\mathcal{A}_j}}}\otimes \ket{j}\bra{j}$ and $\mathcal{O}\big(\frac{1}{\epsilon_\text{gap}^3}\big(\log\frac{m}{\delta}\big)^2\big)$ other gates.

In the following subsection, we derive the gate complexity of the projectors for the scenario in which the Gibbs state $\rho_\tau$ is replaced with a TPQ state. The OR lemma combined with binary search enables one to check the $m+1$ constraints and search for a broken one if it exists. 

 
\subsubsection{OR Lemma with TPQ states}
\label{sec:constructing projectors}

In order to combine the projectors $\Pi_j$ [Eq.~\eqref{eq:amp_proj}] in the OR lemma with our TPQ states $\ket{\tilde{\psi}}$, we create the state
\begin{equation}
    \tilde{\rho}
    := 
    \bigotimes_{i=0}^{x-1} 
    \big(\ket{\tilde{\psi}_{\tau}}\bra{\tilde{\psi}_{\tau}}
    \otimes \ket{0^\ell}\bra{0^\ell}\big), 
\end{equation} 
as the input state.\footnote{For simplicity in our proof we consider the scenario in which each copy is an identical TPQ state, i.e. each TPQ state is generated with the same random circuit $U$. One might expect that using a different random circuit for each TPQ state could lead to a smaller error in practice.} 
To incorporate the error, $\xi$, that stems from the use of TPQ states (Definition~\ref{def:TPQ_state}), we set $\epsilon_\text{gap}=\xi$ and $x=\frac{1}{\xi^2}\ln\frac{m+1}{\delta}$ as discussed in Appendix~\ref{app:copies}.
Under this setup, in case~$(i)$ of Lemma~\ref{lemma:fastOR_QET} we have 
\begin{align}
\label{eq:projector_case1}
\begin{split}
    &\tr[(\tilde{\rho}\otimes\ket{0^y}\bra{0^y})\Pi_j]
    \ge
    1-3\delta/(m+1)-C_\xi e^{-\alpha n},
\end{split}
\end{align}
and in case $(ii)$ we have 
\begin{align}
\label{eq:projector_case2}
\begin{split}
    &\tr[(\tilde{\rho}\otimes\ket{0^y}\bra{0^y})\Pi_j]
    \le 
    2\delta/(m+1) + C_\xi e^{-\alpha n},
\end{split}
\end{align}
which can be shown using Lemma~\ref{lemma:eigenvalue_disc_density} (see Appendix~\ref{app:projector_inequality} for the derivation).

Hence, the set of operators $\{\Pi_j\}_{j=0}^{m}$ serve as projectors in Lemma~\ref{lemma:fastOR_QET} with $\upsilon = 3\delta/(m+1) + C_\xi e^{-\alpha n}$ and $\phi = 2\delta/(m+1)+ C_\xi e^{-\alpha n}$.
The lemma tells us that in case $(i)$, the algorithm accepts with probability at least $P_1:=(1-3\delta/(m+1)-C_\xi e^{-\alpha n})^2/4-\zeta$, and in case $(ii)$ the algorithm accepts with probability at most $P_2:=10\delta+ 5(m+1)C_\xi e^{-\alpha n}+\zeta$. 

To increase the probability of successfully distinguishing the two cases, we amplify the gap $P_1-P_2$ by repeating the test $K$ times. 
We assume $P_1-P_2\ge \Omega(1)$ by taking $\delta$ and $\zeta$ to be sufficiently small.\footnote{The probabilities $P_1$ and $P_2$ depend on the TPQ success probability $C_\xi e^{-\alpha n}$, which we do not have precise control over. We carefully consider this in  Appendix~\ref{app:amplify_OR_gap}.}
For instance, if we set $\delta = 1/184$ and $\zeta = 1/32$, then $P_1-P_2\ge1/8-\calO(mC_\xi e^{-\alpha n})$.
We let $X_i = 1$ if the $i^\mathrm{th}$ repetition of the test accepts, and $X_i = 0$ if it does not. 
Given the count of acceptances $\sum_{i=1}^{K}X_i$, we conclude that there is a broken constraint if $\sum_{i=1}^K X_i \geq K\frac{P_1+P_2}{2}$, and all the constraints are satisfied if $\sum_{i=1}^K X_i \leq K\frac{P_1+P_2}{2}$. 
Setting $K = \calO\big(\ln\frac{\log_2m}{\tilde{\delta}}\big)$ and making use of the Chernoff bound, we find the test yields a correct outcome with probability at least $1-\tilde{\delta}^{1-\calO(mC_\xi e^{-\alpha n})}/\log_2m$ (see Appendix~\ref{app:amplify_OR_gap}).

Case $(i)$ in the OR Lemma only tells us if there exists a broken constraint, not which index it has. In order to search for the index, we perform a binary search by repeating the test $\log_2 m$ times on different subsets of the constraints. We then have an index $j$ such that $\tr\left[\rho_\tau \mathcal{A}_j\right]\geq 1/2+\epsilon/4$ with probability at least $1-\tilde{\delta}^{1-\calO(mC_\xi e^{-\alpha n})}$ (Lemma 15 in~\cite{aaronson2018shadow}).
Thus, in total we require $K\times\log_2 m=\calO(\log m \log\frac{\ln m}{\tilde{\delta}})$ repetitions of the test. Together with the complexity of implementing the OR lemma, we end up with the following theorem for the resource requirements of the broken constraint check.

\begin{theorem}[Broken constraint check with gap promise]
\label{thm:broken_constraint}
    Let $\epsilon, \xi$ and $\tilde{\delta}\in(0,1)$.
    Also, let $\{A_j\}_{j=0}^{m}$ be a set of constraint matrices, and $\ket{\tilde{\psi}_{\tau}}$ be a TPQ state [Eq.~\eqref{eq:approx_TPQ}] corresponding to the Gibbs state at the $\tau^\mathrm{th}$ step of the quantized MMW algorithm with 
    $\mathrm{Pr}_{U}[|\bra{\tilde{\psi}_\tau}A_j\ket{\tilde{\psi}_\tau}-\tr[\rho_\tau A_j]|\geq \xi]\leq C_\xi e^{-\alpha n}$ 
    for all $j=0,\dots,m$. 
    Suppose either
    \begin{enumerate}
        \item[(i)] there exists a $j$ such that $\tr[\rho_\tau A_j] \ge b_j + \epsilon+2\xi$, or
        \item[(ii)] $\tr[\rho_\tau A_j] \le b_j + \epsilon$ for all $j$,
    \end{enumerate}
    is promised to hold.
    Then there is an algorithm that in case $(i)$ outputs an index $j$ such that $\tr[\rho_\tau A_j]\geq b_j+\epsilon+2\xi$, or in case $(ii)$ correctly concludes $\tr[\rho_\tau A_j] \le b_j+\epsilon$ for all $j$ with failure probability at most $\tilde{\delta}^{1-\calO(mC_\xi e^{-\alpha n})}$. 
    
    The test makes
    \begin{equation}
        \mathcal{O}\left(
        \frac{1}{\xi^2}\sqrt{m}(\log m)^2
        \Big(\log\log m + \log\frac{1}{\tilde{\delta}}\Big)
        \right)
    \end{equation}
    uses of the operator $\sum_{j=0}^m U_{\sqrt{\mathcal{A}_j}}\otimes \ket{j}\bra{j}$,
    \begin{equation}
        \calO\left(
        \frac{1}{\xi^3}\sqrt{m}(\log m)^3
        \Big(\log\log m + \log\frac{1}{\tilde{\delta}}\Big)
        \right)
    \end{equation}
    other gates,
    and $\calO\big(\frac{\log m}{\xi^2}\big)$ ancillary qubits.
\end{theorem}


\subsection{Total complexity of the TPQ-SDP solver}
\label{sec:errors}

Putting all the protocols discussed so far together, we discuss the complexity of the entire quantum SDP protocol using TPQ states (Algorithm~\ref{alg:TPQ_SDP}).

The broken-constraint search assumes a probability gap between two cases as stated in Theorem~\ref{thm:broken_constraint}. This gap turns into another source of error. 
If the algorithm returns an index $j^*$, then we know with probability at least $1-\tilde{\delta}^{1-\calO(mC_\xi e^{-\alpha n})}$ that $\tr[\rho_\tau A_{j^*}]>b_{j^*}+\epsilon$. If the algorithm returns no index, then we know with the same confidence that $\tr[\rho_\tau A_j]<b_j+\epsilon+2\xi$ for all $j=0,\dots,m$. 
Therefore, if the broken-constraint-check primitive keeps returning an index until $T = \frac{8\ln N}{\epsilon^2}$ steps of the Zero-Sum algorithm, then we can conclude that $\mathcal{S}_{\epsilon} = \emptyset$ as defined in Eq.~\eqref{eq:set of satisfying matrices} with probability at least $1-T\tilde{\delta}^{1-\calO(mC_\xi e^{-\alpha n})}$.
On the other hand, if the primitive does not return index at step $\tau$ then we can conclude that $\rho_\tau\in\mathcal{S}_{\epsilon+2\xi}$ with probability at least $1-\tau\tilde{\delta}\geq 1-T\tilde{\delta}^{1-\calO(mC_\xi e^{-\alpha n})}$. 

The final step is to ensure the probability of any error in checking for broken constraints to be low enough such that over the $T\log\frac{1}{\epsilon} = \frac{8}{\epsilon^2}\ln N\log\frac{1}{\epsilon}$ broken constraint checks with high probability no failure occurs. 
Setting $\xi=\epsilon$, we arrive at the following theorem:

\vspace{2mm}
\begin{theorem}[Complexity of the TPQ-SDP solver]
\label{thm:total_complexity}
    Let $\cal{T}_{K^\tau}$ be the gate complexity of block-encoding of the shifted Hamiltonian $K_\tau$ and $\cal{T}_{\sqrt{\calA}}$ be the gate complexity of the operation $\sum_{j=0}^{m}U_{\sqrt{\mathcal{A}_j}}\otimes \ket{j}\bra{j}$. 
    Assume that the maximum  purity over all $\tau$ of the Gibbs state $\rho_\tau=e^{-\beta^\tau H^\tau}/\tr[e^{-\beta^\tau H^\tau}]$ decays exponentially with $n=\log{N}$ as $\tr[\rho_\tau^2]\propto e^{-\alpha n}$ for some constant $\alpha$. 
    Then, we can find the solution of the primal SDP within an error $4R\epsilon$ with probability at least $1-T\tilde{\delta}^{1-\calO(C_\epsilon me^{-\alpha n})}\log\frac{1}{\epsilon}$. 
    The algorithm operates with the gate complexity
    \begin{equation}
        \mathcal{O}\left(
        \left(
        \frac{\sqrt{N}}{\epsilon^{9/2}}\mathcal{T}_{K^\tau}
        +\sqrt{m}\left(
        \frac{\mathcal{T}_{\sqrt{\mathcal{A}}}}{\epsilon^4}+\frac{1}{\epsilon^5}
        \right)
        \right)
        \cdot
        \mathrm{polylog}\left(N, m, \frac{1}{\epsilon}, \frac{1}{\tilde{\delta}}\right)
        \right),
    \end{equation}
    and maximum circuit depth
    \begin{equation}
        \mathcal{O}\left(
        \left(\frac{\sqrt{N}}{\epsilon^{1/2}}\mathcal{T}_{K^\tau}+\sqrt{m}\left(\mathcal{T}_{\sqrt{\mathcal{A}}}+\frac{1}{\epsilon^3}\right)\right)
        \cdot
        \mathrm{polylog}\left(N, m, \frac{1}{\epsilon}, \frac{1}{\tilde{\delta}}\right)
        \right).
    \end{equation}
\end{theorem}

 
\section{Application of the TPQ-SDP solver: Hamiltonian learning}
\label{sec:numerics}

We verify the efficiency of our TPQ-SDP algorithm by applying it to the Hamiltonian learning problem. In Hamiltonian learning, one attempts to infer a Hamiltonian from experiments on some physical system~\cite{Wiebe_2014, Kokail2021, dutkiewicz2023advantage}. 
Specifically, we consider the setting in which the target system is assumed to be in some thermal equilibrium state $\rho_{\mathrm{target}}=e^{-\beta H_{\theta_{\mathrm{target}}}}/Z$, where $H_{\theta_\mathrm{target}} = \sum_{j=1}^m \theta_{ \mathrm{target},j} O_j$. We have access to expectation values $\tr[\rho_{\mathrm{target}}O_j]$ for $j=1,\cdots, m$~\cite{Anshu_2021, Haah_2022} and are tasked with learning the Hamiltonian $H_\theta=\sum_{j=1}^{m} \theta_j O_j$ such that $\lVert\theta-\theta_{\mathrm{target}}\rVert\leq \epsilon'$ with $\theta:=(\theta_1,\dots,\theta_m)^T$ and $\theta_\text{target}:=(\theta_{\text{target},1},\dots,\theta_{\text{target},m})^T$. Here $\lVert\cdot\rVert$ could be any vector norm, but is typically chosen to be the L2 or L-$\infty$ norm.

This problem arises in the context of the verification and certification of quantum technologies, for example, when one wishes to verify that one has implemented a particular Hamiltonian, or Gibbs state, correctly on a quantum device~\cite{Wang_2017}. This is generally hard due to the exponentially large Hilbert space. In addition, it addresses the following fundamental questions: 1) can we validate our theoretical models of physical systems within a reasonable amount of time, and 2) how much can we ever learn about nature?

We can massage the Hamiltonian learning problem into our SDP framework of Eq.~\eqref{eq:sdpconstr}. Setting the target expectation values to $b_j= \tr[\rho_{\mathrm{target}}O_j]$, we can define the following feasibility problem \begin{equation}
\label{eq:hamlearn}
\begin{aligned}
    b_j-\epsilon \leq \tr [O_j \rho]\leq b_j+\epsilon,
    & \quad
    \forall j = 1,\dots,m
    \\
    \rho \succeq 0,
    & \quad \tr{\rho}=1.
\end{aligned}
\end{equation} 
This can be solved directly with the quantized MMW method (Algorithm~\ref{alg:TPQ_SDP}). Note that, since we do not have an objective function $L$, we do not need to use binary search. 

An important parameter in the Hamiltonian learning feasibility problem is the error~$\epsilon$. In general, knowing that the expectation values are within the error $\epsilon$ does not imply that the parameters are close to the target as well, i.e., $\lVert\theta-\theta_{\mathrm{target}}\rVert\leq \chi\neq \epsilon$. Therefore, in order to solve the Hamiltonian learning problem within a certain error $\chi$, one needs to have a relation between the difference $\epsilon$ in expectation values and the difference $\chi$ in parameters.
This problem has been extensively studied, and for a specific family of geometrically local target Hamiltonians such a relation has been derived by showing that the relative entropy is strongly convex~\cite{Anshu_2021, Haah_2022}. 
The target Hamiltonians we use in our numerical experiments in Section~\ref{sec:numerics} satisfy these assumptions, and we can relate $\epsilon$ to $\chi$. For more general cases, one can aim to make $\epsilon$ small such that the problem can still be solved in a reasonable amount of time. In particular, for $\epsilon\to 0$ one can recover the exact parameters of any target Hamiltonian. 

We solve the feasibility problem via classical state-vector simulations of our TPQ-SDP solver. To this end, at step $\tau$, we create a TPQ state $\ket{\tilde{\psi}_\tau}$ for a Hamiltonian $H^\tau$ [Eq.~\eqref{eq:TPQ}] by approximating the matrix exponential $e^{-\beta^\tau H^\tau}$ with a Lanczos diagonalization of $H^\tau$~\cite{Jakli1994, trefethen1997numerical}, and applying it to a random Clifford state $U_{\mathrm{Cliff}}\ket{0}$. We compute the expectation values $\bra{\tilde{\psi}_\tau}O_j\ket{\tilde{\psi}_\tau}$ and average over a few random Clifford states to ensure that with high probability we are close to the exact Gibbs state expectation values (Definition~\ref{def:TPQ_state}). We then check the constraints and update the parameter vector $\theta^\tau$, following our quantized MMW Algorithm~\ref{alg:TPQ_SDP}. This process is repeated until a feasible Gibbs state $\rho$ is found. Due to the setup of the problem, we will always find a feasible Gibbs state.
The classical simulation complexity of this method is lower bounded by $\tilde{\mathcal{O}}(N+m)$ for sparse input Hamiltonians $H^\tau$, which matches the theoretical lower bound for solving generic SDPs classically\footnote{Note that this bound is for generic sparse Hamiltonians. One may improve this bound greatly if e.g. the terms in the Hamiltonian have low support~\cite{Haah_2022,bakshi2023learning}.}. These simulations are sufficient to validate our approach. In particular, it allows us to show that TPQ states can be used to provide approximate solutions for SDPs. In practice, one would run the TPQ-SDP solver on a quantum device.
The data shown in Fig.~\ref{fig:Numdata} is obtained by calculating the true Gibbs state expectation values using exact diagonalization, but all updates to $H^\tau$ were calculated using simulations of the TPQ states.

In the first numerical experiment, we aim to learn the two-dimensional (2D) spinless Hubbard model, which is extensively studied in condensed matter physics~\cite{Hubbard1963, Arovas_2022}. 
This model is described by the Hamiltonian,
\begin{equation}
\label{eq:spinlesshubbard}
    H_{\mathrm{Hubbard}} 
    = \sum_{i=1}^{n_x\times n_y}\mu 
    \Big(c^\dagger_ic_i - \frac{1}{2}\Big) 
    + \sum_{\langle i,j\rangle} 
    \omega(c^\dagger_i c_j - c_ic_j^\dagger) 
    + U\Big(c^\dagger_ic_i-\frac{1}{2}\Big)\Big(c^\dagger_jc_j-\frac{1}{2}\Big),
\end{equation} 
where $\langle i,j\rangle$ is the set of nearest neighbour indices on a 2D square lattice of size $(n_x\times n_y)$ with open boundary conditions. The spinless Fermionic creation and annihilation operators, $c^\dagger_i$ and $c_i$, satisfy $\{c_i^\dagger, c_j\}=\delta_{ij}$ and $\{c_i, c_j\}=0$, and each term in the sums is orthogonal to each other. Note that mutual orthogonality is not necessary for the MMW method to succeed, but reduces $\lVert\theta-\theta_{\mathrm{target}}\rVert$ for a given error in the expectation values~\cite{Anshu_2021}.
The parameter $\mu$ is the chemical potential that describes the density of fermions, $w$ the nearest-neighbour hopping strength, and $U$ the nearest-neighbour density-density interaction strength. In our study we set these parameters to $\mu=1.0$, $\omega=0.5$, and $U=1.2$. We define the target state $\rho_{\mathrm{target}}=e^{-\beta H_{\mathrm{Hubbard}}}/Z$ with $\beta=0.4$ to ensure the purity $\tr[\rho_\mathrm{target}^2]$ is small enough that any deviations in the TPQ state expectation value from the true Gibbs state expectation value are with high probability below our desired target error $\epsilon=0.05$.

\begin{figure}[h!]
\centering
 \includegraphics[width=0.95\linewidth]{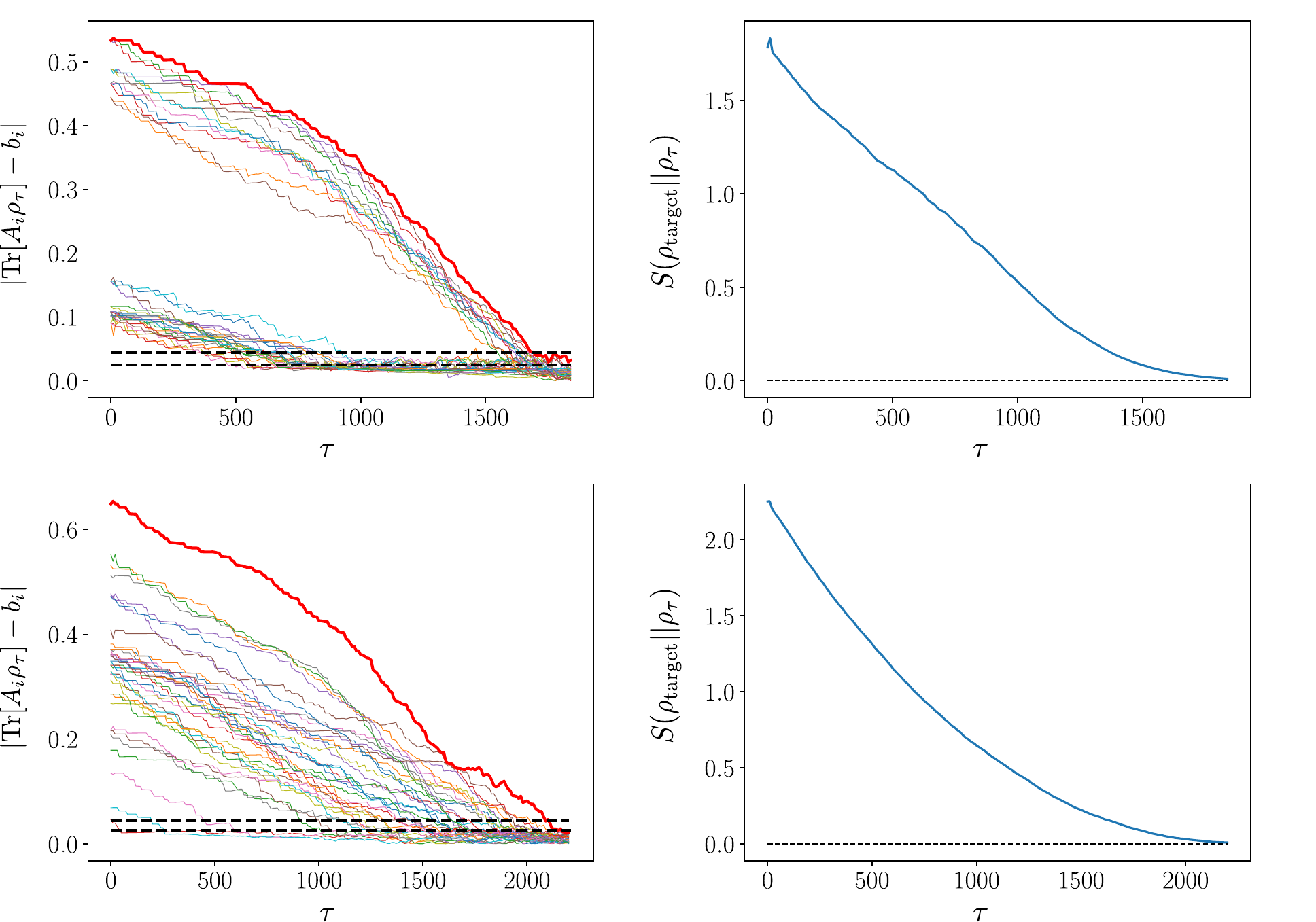}
 \caption{Numerical simulation results of the TPQ-SDP solver applied to the Hamiltonian learning problem. In the top row the target Hamiltonian is the 2D Hubbard model [Eq.~\eqref{eq:spinlesshubbard}] and in the bottom row the 1D XXZ Heisenberg model [Eq.~\eqref{eq:xxz}]. The left panels show the errors in the expectation values as a function the Zero-Sum game update step $\tau$. The maximum error is highlighted in red, and the SDP optimization is stopped once the target SDP accuracy of $\epsilon=0.05$ is reached (bottom dashed black lines). The other dashed line, shows the SDP plus the maximum TPQ error, which is slightly larger. The right plots show the relative entropy $S(\rho_{\mathrm{target}}\Vert\rho_\tau)$ between the target Gibbs state and the Gibbs state of the learned model Hamiltonian. For these simulations we used a system of $n=10$ qubits, and the TPQ state expectation values were computed from a median of means of 3 subsets of 25 TPQ states. At each update step $\tau$ the $m$ constraints were checked in a random order.} 
\label{fig:Numdata}
 \end{figure}
 
In the top row of Fig.~\ref{fig:Numdata}, we show the results of our TPQ-SDP solver for the Hubbard model target Hamiltonian of $n=5\times2=10$ lattice sites. In the left panel, we observe that the difference in the expectation values between the model $\tr[A_j\rho]$ and target $b_j$ gradually decreases during the optimization. After $\tau\approx1700$ iterations the target SDP accuracy of $\epsilon=0.05$ is reached. This demonstrates that our method is able to solve the Hamiltonian learning feasibility problem up to the TPQ error. In the right panel, we see that simultaneously the relative entropy $S(\rho_{\mathrm{target}}\Vert\rho)$ between the target and model states decreases monotonically to zero. This shows that the MMW method can be seen as a relative entropy minimization if there exists a state which satisfies all the constraints. 

At the end of the optimization, the mean-squared error in the $m$ parameters is $\frac{1}{m}\lVert\theta_{\mathrm{target}}-\theta\rVert^2_2 = 0.020$. The maximum multiplicative error in the learned parameters is $0.21$, which is an order of magnitude smaller than the size of the parameters $\mathcal{O}(10^0)$. As discussed above, this can be further improved by choosing a smaller $\epsilon$ for the SDP target accuracy, which incurs a higher computational cost due to a higher number of SDP iterations. Note, however, that we can not set $\epsilon$ to an arbitrary small value because we require $1-\mathcal{O}(\tr[\rho_{\mathrm{target}}^2]/\epsilon^2)>0$ for a positive success probability (recall Theorem~\ref{thm:total_complexity}). 
For geometrically local models, like the Hubbard model, this is not a problem since there exists a bound $\lVert\theta_{\mathrm{target}}-\theta^T\rVert\leq e^{-\mathcal{O}(\beta)}\mathrm{poly}(n)\epsilon$, which scales as a polynomial of the number of qubits, $n$~\cite{Anshu_2021}. This means that if we set $\epsilon \leq e^{-\beta}/\mathrm{poly(n)}$, we can find the parameters to arbitrary precision when $n\to\infty$. At the same time, $\mathcal{O}(\tr[\rho_{\mathrm{target}}^2]/\epsilon^2)$ decreases exponentially with $n$ for an exponentially vanishing purity of the target, so the success probability of our algorithm remains finite and positive. 

In the next numerical experiment, we focus on a one-dimensional Heisenberg XXZ model~\cite{Heisenberg1928, Franchini_2017} in a disordered external magnetic field, which has Hamiltonian \begin{equation}\label{eq:xxz}
    H_{\mathrm{XXZ}} =\sum_{i=1}^{n-1}J\left(\sigma^{x}_i\sigma^{x}_{i+1}+\sigma^y_i\sigma^y_{i+1}\right)+\Delta \sigma_i^z\sigma^z_{i+1} + \sum_{i=1}^{n}h^z_i \sigma^z_i.
\end{equation} 
Here $J$ is the $XY$ nearest-neighbour coupling strength, $\Delta$ the $Z$ coupling strength, and $h^z_i$ a random local magnetic field in the $z$ direction. 
Due to the disordered fields $\{h_i^z\}$ there are a few more parameters that we need to learn to a high precision. As target Hamiltonian we use $J=1.0$, $\Delta=0.5$, and $h_i^z$ are drawn from independent uniform distributions over~$[-2,2]$. We set the temperature to $\beta=0.4$ and define the target state $\rho_{\mathrm{target}}=e^{-\beta H_{\mathrm{XXZ}}}/\tr{}[e^{-\beta H_{\mathrm{XXZ}}}]$. 

In the bottom row of Fig.~\ref{fig:Numdata}, we show the performance of our TPQ-SDP solver for the XXZ target Hamiltonian. Similarly to the Hubbard model, both the error in the expectation values and the relative entropy decrease during the optimization process. This time, more iterations are required before the target accuracy $\epsilon=0.05$ is reached, which is after $\tau\approx 2200$ steps.
The mean-squared error in the Hamiltonian parameters is $\frac{1}{m}\lVert\theta_{\mathrm{target}}-\theta\rVert^2_2 = 0.015$ when the optimization completes.

 
\section{Discussion and Conclusion}
\label{sec:conclusion}

In this work, we provide a novel quantum algorithm for solving primal SDPs of $N\times N$ variables and $m$ constraints up to error $\varepsilon$. The algorithm is based on a quantization of the MMW algorithm and uses TPQ states to estimate Gibbs-state expectation values. We proved that the algorithm requires a maximum of  
\begin{equation}
\label{eq:finalcomplex}
        \mathcal{O}\left(
        \left(
        \sqrt{N}\left(\frac{R}{\varepsilon}\right)^{9/2}\mathcal{T}_{K_\tau}
        +\sqrt{m}\left(
        \mathcal{T}_{\sqrt{\mathcal{A}}}\left(\frac{R}{\varepsilon}\right)^4+\left(\frac{R}{\varepsilon}\right)^5
        \right)
        \right)
        \cdot
        \mathrm{polylog}\left(N, m, \frac{R}{\varepsilon}, \frac{1}{\tilde{\delta}}\right)
        \right),
\end{equation}
quantum gates, and succeeds with probability $1-T\tilde{\delta}^{1-\calO(mC_{\varepsilon/R} e^{-\alpha n})}\log\frac{R}{\varepsilon}$. Here $R$ is an upper bound on the trace of the optimal value of the SDP, $\mathcal{T}_{K^\tau}$ is the gate complexity of block encoding the shifted Hamiltonian $K^\tau$, and $\mathcal{T}_{\sqrt{\mathcal{A}}}$ the gate complexity of block encoding the square root of the shifted constraint matrices $\mathcal{A}$. The classical simulation time of the algorithm is lower bounded by $\tilde{\mathcal{O}}(N+m)$. 

Compared to the best previous quantum SDP solver in Ref.~\cite{vanapeldoorn2019}, which has gate complexity 
\begin{equation}
     \mathcal{O}\left(\left(\left(\sqrt{m}+\sqrt{N}\frac{R}{\varepsilon}\right)s\left(\frac{R}{\varepsilon}\right)^4\right)\cdot\mathrm{polylog}\left(N, m, \frac{R}{\varepsilon}, \frac{1}{\tilde{\delta}}\right)\right),
 \end{equation} 
where $s$ is an upper bound on the row-sparsity of the input Hamiltonian. Our TPQ-SDP solver reduces the power of $\varepsilon$ multiplying $\sqrt{N}$ by a factor $\varepsilon^{-1/2}$. This is due to our assumption of being able to block encode $K^\tau$. If such a block encoding is unavailable, we get an additional factor $\varepsilon^{-1}$ multiplying $\sqrt{N}$. Note that the power of $\varepsilon$ multiplying the $\sqrt{m}$ in our algorithm is worse, but this is because we provide an explicit circuit implementation of the quantum OR lemma which van Apeldoorn et. al. neglect. To our knowledge, this extra factor of $\varepsilon^{-1}$ is necessary when constructing the projectors. Furthermore, our algorithm does not require purifying the Gibbs states in the MMW algorithm, and hence reduces the required number of qubits. For SDPs like the Hamiltonian learning problem this results in an exponential saving in the number of ancillary qubits. This makes our algorithm more useful when one only has access to quantum devices with a limited number of qubits. 

An important aspect of our algorithm is that TPQ states only approximate Gibbs state expectation values up to a finite error, with a probability that depends on the purity of the Gibbs state. For this reason, the success probability of the algorithm depends on $\tilde{\delta}^{1-\mathcal{O}(m C_{\varepsilon}e^{-\alpha \log{N}})}$, where $C_{\varepsilon}e^{-\alpha \log{N}}$ is an upper-bound on the purity of the Gibbs states. In other words, we assume the purity vanishes when the size of the SDP is increased. We argue that this is true for physical Hamiltonians with extensive free energy. In addition, we prove that a specific spectral condition (high density of low energy states, Proposition~\ref{prop:sdp_pur_cond}) guarantees a vanishing purity, and show that random matrices from the Generalized Unitary Ensemble satisfy it. We provide numerical evidence that this condition holds for the one-dimensional Heisenberg XXZ model and the two-dimensional Hubbard model. In future work, it would be interesting to investigate this condition further and see to which type of Hamiltonians our TPQ-SDP solver can be applied.   

Investigating the spectral condition is also interesting with respect to the prospect of proving a potential larger quantum speedup of our algorithm. We show that when the spectral condition is satisfied, the number of rounds of amplitude amplification for the TPQ state preparation scales with $N^{1/4}$ instead of $N^{1/2}$. This shifts the computational bottleneck to finding the ground state energy, which only needs to be performed once per step in the MMW algorithm. Therefore, we see an $\calO(\left(\frac{\varepsilon}{R}\right)^2)$ saving in the asymptotic block-encoding querey complexity. Further, if one can find the ground state energy to high enough accuracy in a time that scales at maximum with $N^{1/4}$, this may imply a quartic speedup for solving SDPs with Hamiltonians that satisfy the condition. For example, in~\cite{Chen2023} the authors argue that ground state energies of sparse random quantum Hamiltonians can be found in $\text{Poly}(n)$ depth.

In order to make this rigorous, one needs to prove a classical lower bound on the complexity in the presence of the condition, and one needs to investigate if the condition can reduce the ground energy estimation complexity. Note that our result is not specific to our TPQ-state approach but also holds for quantum SDP algorithms based on purified Gibbs states, although TPQ states are special in that they require a high density of low energy states to give a good approximation to the Gibbs state expectation values, whereas other Gibbs samplers perform better if the Hamiltonian is gapped i.e. has a low density of low energy states, for example the Gibbs sampler used in~\cite{vanApeldoorn2020quantumsdpsolvers}.

\begin{table}
\begin{center}
\def\arraystretch{1.5}
\begin{tabular}{|c||c|c|c|c |} 
 \hline
 \multirow{2}{*}{\makecell{\textbf{System size} \\ $n_x\times n_y$}} 
 & \multicolumn{2}{c|}{\textbf{TPQ State preparation}}
 & \multicolumn{2}{c|}{\qquad\textbf{Proof-of-concept} \enspace\enspace}
 \\
 \cline{2-5}
 & Toffoli gates & qubits & Toffoli gates & qubits
 \\
 \hline\hline   
 $2\times 2$   
 & $3.28\times10^{7}$ 
 & 18 
 & $1.61\times10^{6}$ 
 & 13 
 \\
 \hline 
 $4\times 4$     
 & $4.62\times10^{10}$
 & 46
 & $1.22\times10^{8}$
 & 29
 \\
 \hline
 $6\times 6$   
 & $3.18\times10^{14}$
 & 90
 & $1.45\times10^{9}$
 & 53
 \\[.5ex] 
 \hline
\end{tabular}
\caption{Resource requirements (Toffoli gates and qubits) of the TPQ state preparation routine for various instances of the Hamiltonian learning problem with $\epsilon=0.05$. The target Hamiltonian is the spinless Hubbard model on a squared lattice of size $n_x\times n_y$. The second column shows the resources for preparation of TPQ states including amplitude amplification to boost the success probability. Asymptotically this has a quadratic speedup with respect to the size of the SDP $N$. The right column shows the resource requirements for a proof-of-concept preparation of a TPQ state without amplitude amplification. This has no theoretical quantum advantage since it requires an unknown amount of post-selection.}
 \label{tbl:resources}
\end{center}
\end{table}

Such a speedup would also reduce the quantum resource requirements for the application of our TPQ-SDP solver to practically relevant problems. In Table~\ref{tbl:resources}, we provide the requirements (number of qubits, and the maximum number of gates in one circuit) of the TPQ state preparation routine of our algorithm applied to learning a two-dimensional Hubbard model of several different system sizes. Instead of checking the constraints on the quantum computer, one could measure the expectation values of these TPQ states with for example classical shadow tomography~\cite{Huang2020, Coopmans_2023}, and check the $m$ constraints on a classical device. This approach has a quadratic speedup with respect to $N$. In the second column of the table, we observe that we require order $10^{14}$ Toffoli gates for a system of $n=36$ qubits. If one could place a tighter bound on the partition function, e.g. by showing the spectral condition applies to this model, this could potentially be reduced to order $10^{11}$ Toffoli gates, which is three orders of magnitude smaller. Moreover, one could reduce the gate count further by trading depth for additional ancilla qubits. For example, by exploiting a different block-encoding method. 

In the fourth and fifth column of Table~\ref{tbl:resources}, we show the hardware requirements if one only cares about a proof-of-concept demonstration, without a theoretically guaranteed speedup. For this, one could prepare the TPQ state by post-selecting on the ancilla qubits being ``00\dots 0'' without amplitude amplification. In this case we require order $10^9$ Toffoli gates for a system of 36 qubits, which is far out of reach for the capabilities of the current generation of quantum devices. However, note that our bounds are derived for generic cases and do not take into account specifics of the Hubbard model. Circuit compilation techniques, and using the spectral properties of the model may significantly reduce the resource requirements. Studying this, and also investigating larger speedups stemming from the spectral condition, is particularly important in light of demonstrating quantum speedups in practice~\cite{Campbell_2019, Babbush_2021}. 

A last important avenue for future work is investigating which types of SDPs our TPQ-SDP solver, and more generally other quantum SDP solvers, can provide a genuine quantum speedup. Our work shows that Hamiltonian learning of quantum many-body Hamiltonians may provide such an example. Reference~\cite{vanapeldoorn2019} argues that shadow tomography, quantum state discrimination, and E-optimal design are other example problems that can be addressed with quantum SDP algorithms. Most of these algorithms have in common that the bound on the optimal solution of the SDP, $R$, does not scale with the problem size $N$, and the rescaling of the SDP does not form a problem. For combinatorial optimization problems, such as the maximum cut problem~\cite{Goemans1995}, this is not the case, and one does not obtain a genuine quantum speedup with most quantum SDP algorithms without access to a QRAM~\cite{Augustino2023}. Finding solutions to this problem or finding other useful examples, is therefore critical for the future success of quantized SDP solvers. 


\section*{Acknowledgments}

We thank David Amaro, Eric Brunner, Jingjing Cui, Samuel Duffield, Mattia Fiorentini, Alexandre Krajenbrink, Enrico Rinaldi, and Matthias Rosenkranz for stimulating discussions. We thank Marcello Benedetti, Steven Herbert and Tuomas Laakkonen for feedback on an earlier version of this manuscript. 
\appendix

\newpage
\appendix

\renewcommand\thefigure{\thesection \arabic{figure}}
\renewcommand\thetable{\thesection \arabic{table}} 
\begin{center}
	\noindent\textbf{Appendix}\\
	\bigskip
	\noindent\textbf{\large{Quantum Semidefinite Programming with Thermal Pure Quantum States}}
\end{center}
\etocdepthtag.toc{mtappendix}
\etocsettagdepth{mtchapter}{none}
\etocsettagdepth{mtappendix}{subsection}
\tableofcontents
\newpage

\section{Approximate Jaynes Principle}
\label{app:Jaynes}
In this appendix, we review the proof of the approximate Jaynes principle~\cite{Jaynes1957}, from which the convergence of the Zero-Sum approach outlined in the main text follows.

The approximate Jaynes principle, see Lemma 4.6 in Ref.~\cite{Lee2015}, states the following.
\begin{theorem}[Restatement of Lemma 4.6 in Ref.~\cite{Lee2015}]
\label{thm:approx_jaynes}
For every $\epsilon>0$, the following holds. Let $\mathcal{T}$ be a finite set of Hermitian matrices labelled $\{O_j\}$, $i = 1,\dots,m$,\footnote{While a finite set of Hermitian matrices suffice for our quantum SDP algorithm, the statement also holds for any compact set of Hermitian matrices.} and let $\eta$, $\rho_0\in\mathcal{D(H)}$ be density matrices. There exists $O_{j(1)},O_{j(2)},\dots,O_{j(T)}\in\mathcal{T}$ such that
\begin{equation}
\label{form}
    \tilde{\rho} 
    := \frac{\exp\left(\ln \rho_0-\frac{\epsilon}{4\Delta(\mathcal{T})^2}\sum_{t=1}^{T}O_{j(t)}\right)}{\tr\left[\exp\left(\ln \rho_0-\frac{\epsilon}{4\Delta(\mathcal{T})^2}\sum_{t=1}^{T}O_{j(t)}\right)\right]}
    \in\mathcal{D(H)}
\end{equation}
satisfies\footnote{Equation~\eqref{eq:approx_expectations} is corrected from~\cite{Lee2015} which had a sign error.}
\begin{equation}
\label{eq:approx_expectations}
    \tr\left[O(\tilde{\rho}-\eta)\right]
    \leq\epsilon,\quad\forall O\in\mathcal{T},
\end{equation}
where $T\leq\lceil\frac{8}{\epsilon^2}S(\eta\Vert\rho_0)\Delta(\mathcal{T})^2\rceil$, $\Delta(\mathcal{T}) = \max_{O\in\mathcal{T}}\lVert O\rVert$, and the relative entropy between $\rho$ and $\rho_0$ is defined by 
\begin{equation}
    S(\eta\Vert\rho_0):=\tr[\eta(\ln\eta-\ln\rho_0)].
\end{equation} 
\end{theorem} 

The theorem states that, if a density matrix $\eta$ with expectation values $\tr[\eta O_j]\leq b_j$ exists for a given set of matrices $\{O_j\}$, $j=1,\dots,m$, then there also exists a Gibbs state $\tilde{\rho}$ with expectation values $\tr[\tilde{\rho}O_j]\leq b_j+\epsilon$ for all $j=1,\dots,m$. Further, the Hamiltonian of the Gibbs state $\tilde{\rho}$ is a linear combination of the $m$ matrices in $\{O_j\}$. This is an extension of the original Jaynes Principle~\cite{Jaynes1957}, which states that if a state $\eta$ with expectation values $\tr[\eta O_j] = b_j$ exists, then there also exists a Gibbs state $\tilde{\rho}$ with expectation values $\tr[\tilde{\rho}O_j]= b_j$, with a potentially infinite inverse temperature $\beta$. Again, the Gibbs state Hamiltonian is some linear combination of the operators $O_j$.
As it becomes clear in the proof, $\rho_0$ serves as an initial guess of the density matrix $\rho$.
For completeness, we adapt here the proof of this statement from Ref.~\cite{Lee2015}. We will make further remarks about $\rho_0$ after the proof.

\begin{proof}
Defining $\rho(\tau) = e^{X(\tau)}/\tr[e^{X(\tau)}]$ for a Hermitian matrix $X(\tau)$ parametrized by $\tau\in\mathbb{R}$, we upper bound the second derivative of the relative entropy between $\rho_\tau$ and the target state $\eta$ as follows:
\begin{equation}
\begin{aligned}
    \frac{d^2}{d\tau^2}S(\eta\Vert\rho(\tau)) 
    &= -\frac{d^2}{d\tau^2}\tr[\eta\log\rho(\tau)]
    \\
    &= -\frac{d^2}{d\tau^2}\tr\left[\eta\left(X(\tau)-\log \tr[e^{X(\tau)}]\right)\right]
    \\
    &= -\tr[\eta X''(\tau)]+\frac{d}{d\tau}\frac{\tr[X'(\tau)e^{X(\tau)}]}{\tr[e^{X(\tau)}]},
\end{aligned}
\end{equation}
where $X':= dX/d\tau$ and $X'':= d^2X/d\tau^2$. Assuming that $X''(\tau)= 0$, which will be the case later on, Lee et al. derived an upper-bound on the second derivative~\cite{Lee2015}
\begin{equation}
\label{eq:rel_entropy_upper}
    \frac{d^2}{d\tau^2}S(\eta\Vert\rho(\tau)) 
    = \tr\left[X'(\tau)\frac{d}{d\tau}\frac{e^{X(\tau)}}{\tr[e^{X(\tau)}]}\right]\leq 2\lVert X'(\tau)\rVert^2. 
\end{equation}
Consider the density matrix
\begin{equation}
    \rho(\tau) 
    = \frac{\exp\left(\ln\rho_0-\int_0^\tau\Lambda_sds\right)}{\tr\left[\exp\left(\ln\rho_0-\int_0^\tau\Lambda_sds\right)\right]}
\end{equation}
for $\tau\geq 0$. We have introduced the matrix $\Lambda_s\in\mathcal{T}$ parametrized by $s\in\mathbb{R}$ and the arbitrary density matrix $\rho_0$ (to keep the logarithm well defined, we restrict to $\rho_0\succ 0$).
Taking the first derivative of the relative entropy $S(\rho\Vert\rho(\tau))$, we obtain 
\begin{equation}
\label{eq:first derivative rel entropy}
\begin{aligned}
    \frac{d}{d\tau}S(\eta\Vert\rho(\tau)) 
    &= -\tr\left[\eta\frac{d}{d\tau}\ln\rho(\tau)\right]
    \\
    &= -\tr\left[-\Lambda_\tau\eta-\eta\frac{d}{d\tau}\ln\tr\left[\exp\left(\ln\rho_0-\int_0^\tau\Lambda_sds\right)\right]\right]
    \\
    &= \tr[\Lambda_\tau \eta]+\tr[\eta]\frac{\frac{d}{d\tau}\tr\left[\exp\left(\ln\rho_0-\int_0^\tau\Lambda_sds\right)\right]}{\tr\left[\exp\left(\ln\rho_0-\int_0^\tau\Lambda_sds\right)\right]}
    \\
    &= \tr[\Lambda_\tau(\eta-\rho(\tau))].
\end{aligned} 
\end{equation}
We now define the elements $O_{j(1)},\dots,O_{j(T)}\in\mathcal{T}$ inductively. Define the time $\tau_t := t\frac{\epsilon}{4\Delta(\mathcal{T})^2}$ and $\rho_t:=\rho(\tau_t)$ for each $t\in\{0,1,2,\dots,T\}$. 
Choose a starting density matrix $\rho_0$ and set $\Lambda_0 = 0$. 
If at any point $\tr[O_j(\rho_{t}-\eta)]\leq\epsilon$ for all $j$, i.e., all the constraints are met, then we are done. Otherwise, choose a broken constraint $O_{j(t)}$ such that
\begin{equation}
    \tr[O_{j(t)}(\rho_{t}-\rho)]>\epsilon.
\end{equation}
Setting $\Lambda_\tau = O_{j(t)}$ for $\tau\in(\tau_t,\tau_{t+1}]$. 
For $\tau\in(\tau_t,\tau_{t+1})$, we see that
\begin{equation}
\label{eq:bound on second derivative}
    -\frac{d^2}{d\tau^2}S(\eta\Vert\rho(\tau)) 
    = \frac{d}{d\tau}\tr[\Lambda_\tau(\rho(\tau)-\eta)] 
    = \frac{d}{d\tau}\tr[O_{j(t)}(\rho(\tau)-\eta)]\geq-2\lVert\Lambda_\tau\rVert^2,
\end{equation}
where the last equality follows from Eq.~\eqref{eq:rel_entropy_upper}.
By integrating both sides between $\tau_t$ and $\tau_{t+1}$ of \eqref{eq:bound on second derivative}, we find
\begin{equation}
\begin{aligned}
    \tr[O_{j(t)}(\rho_{t+1}-\eta)]
    &\geq\tr[O_{j(t)}(\rho_{t}-\eta)]-2\lVert\Lambda_\tau\rVert^2(\tau_{t+1}-\tau_t)
    \\
    &\geq\tr[O_{j(t)}(\rho_{t}-\eta)]-\frac{\epsilon}{2}
    \\
    &>\frac{\epsilon}{2}.
\end{aligned}
\end{equation}
Let $\tau_\mathrm{max} = T\epsilon/4\Delta(\mathcal{T})^2$. Suppose we repeat this step $T$ times and fail to find a $\tilde{\rho}$ which satisfies all the conditions~\eqref{eq:approx_expectations}. Then for all $\tau\in[0, \tau_\mathrm{max}]$ we have a $\Lambda_\tau\in\mathcal{T}$ such that
\begin{equation}
    \tr[\Lambda_\tau(\rho_t-\eta)]
    > \frac{\epsilon}{2}
\end{equation}
which, combined with Eq.~\eqref{eq:first derivative rel entropy}, implies that $\frac{d}{d\tau}S(\eta\Vert\rho(\tau))<-\epsilon/2$ for all $\tau\in[0,\tau_\mathrm{max}]$. Therefore, we can write
\begin{equation}
    S(\eta\Vert\rho(\tau_\mathrm{max}))
    < S(\eta\Vert\rho_0)-\frac{\epsilon}{2}\tau_\mathrm{max}
\end{equation}
For $\tau_\mathrm{max}\geq\frac{2}{\epsilon}S(\eta\Vert\rho_0)$, we get a contradiction since the relative entropy between any two density matrices is always greater than or equal to zero. Hence, if a density matrix $\eta$ which satisfies the SDP conditions exists, then we can always find a Gibbs state of the form of equation \eqref{form} with $T\leq\big\lceil\frac{8}{\epsilon^2}S(\eta\Vert\rho_0)\Delta(\mathcal{T})^2\big\rceil$.
\end{proof} 

An immediate consequence of this proof is that if after over $\lceil\frac{8}{\epsilon^2}S(\eta\Vert\rho_0)\Delta(\mathcal{T})^2\big\rceil$ steps one still does not find a density matrix which satisfies all the constraints, then one may conclude that no such density matrix exists.

The proof of Theorem~\ref{thm:approx_jaynes} with our added constraints $-I\preceq A_j\preceq I$ and $\rho_0 = I/N$ gives the Zero-Sum Algorithm~\ref{alg:zerosum}.
Note that if one is not informed about the properties of the target state $\eta$, the best initial state is the maximally mixed state, $\rho_0 = I/N$. Then
\begin{equation}
    S(\eta\Vert\rho_0) 
    = -\tr[\eta\ln\rho_0]-S(\eta) 
    = \ln N-S(\eta)
    \leq \ln N.
\end{equation}
However, a better initialization $\rho_0$ potentially reduces the integer $T$.
For example, if one were able to solve the mean-field Schr\"{o}dinger equation in a Hamiltonian learning problem, it is possible to start with a state $\rho_0$ closer to the target state $\eta$~\cite{coopmans2023sample}. In either case, $S(\eta\Vert\rho_0)$ is upper bounded by $\ln N$. By rescaling, we also know that $\Delta(\mathcal{T})\leq 1$. 

 
\section{Thermal Pure Quantum States with Unitary Designs}
\label{app:tpq_deriv}

In this appendix, we show that the state
\begin{align}
\label{eq:tpq_app}
    \ket{\psi_\beta} 
    = \frac{e^{-\beta H/2}U\ket{0}}{\sqrt{\bra{0}U^\dagger e^{-\beta H}U\ket{0}}},
\end{align} with $U$ a unitary $k(\ge2)$-design [Eq.~\eqref{eq:psi_beta}],
is a thermal pure quantum (TPQ) state (Definition~\ref{def:TPQ_state}) for the Gibbs state $\sigma_\beta = e^{-\beta H}/\tr[e^{-\beta H}]$ if $\tr[\sigma_\beta^2] \propto e^{-\alpha n}$. Here $\beta\in\mathbb{R}^+$ is a positive constant and $H$ an arbitrary Hermitian matrix (Hamiltonian). We build upon the derivation presented in Ref.~\cite{Coopmans_2023}, which was based on Refs.~\cite{Sugiura2013, Lin2020} for a random Haar unitary~$U$. The bounds presented here are, however, fully exact and not just approximate. At the end of this section, by substituting the SDP Hamiltonian $H=H^\tau$ and constraint matrices $A_j$, we arrive at Eq.~\eqref{eq:TPQ_mse} for the mean-squared error in the main text. 

We start by expanding the mean-squared error as,
\begin{equation}
\label{eq:TPQ_mse}
\begin{aligned}
    &\mathbb{E}_U[(\bra{\psi_\beta}T_j\ket{\psi_\beta}-\tr[\sigma_\beta T_j])^2]
    \\
    &= (\mathbb{E}_U[\bra{\psi_\beta}T_j\ket{\psi_\beta}] - \tr[\sigma_\beta T_j])^2
    + \text{Var}_U[\bra{\psi_\beta}T_j\ket{\psi_\beta}].
\end{aligned}
\end{equation}
Our goal is to show that Eq.~\eqref{eq:TPQ_mse} vanishes exponentially with $n$. In what follows, we place an upper-bound on Eq.~\eqref{eq:TPQ_mse}, which is proportional to the purity of the Gibbs state~$\sigma_\beta$. 

For notational convenience, we introduce
\begin{align}
	&f := \bra{\psi_\beta} e^{-\beta H/2}T e^{-\beta H/2}\ket{\psi_\beta},
	\qquad
	g := \bra{\psi_\beta} e^{-\beta H}\ket{\psi_\beta}.
\end{align}
By averaging over the unitary $k$-design, we obtain the means,
\begin{align}
	\bE_U[g] = \frac{\tr[e^{-H}]}{2^n},
	\qquad
	\bE_U[f] = \frac{\tr[T e^{-H}]}{2^n},
\end{align}
and the variances,
\begin{align}
	\mathrm{Var}_U[g] 
	&= 
	\frac{2^n\tr[e^{-2\beta H}]-(\tr[e^{-\beta H}])^2}{2^{2n}(2^n+1)}
    = 
	(\tr[e^{-\beta H}])^2\frac{2^n\tr[\sigma_\beta^2]-1}{2^{2n}(2^n+1)}
    \nonumber\\
    &\le 
	(\tr[e^{-\beta H}])^2\frac{\tr[\sigma_\beta^2]}{2^{2n}},
	\\
	\mathrm{Var}_U[f] 
	&= 
	\frac{2^n\tr[(Te^{-\beta H})^2]-(\tr[T e^{-\beta H}])^2}{2^{2n}(2^n+1)}
	\le
	(\tr[e^{-\beta H}])^2\frac{\tr[(T\sigma_\beta)^2]}{2^{n}(2^n+1)}
	\nonumber\\
	&\le
	(\tr[e^{-\beta H}])^2\frac{\tr[\sigma_\beta^2]\frac{\tr[(T e^{-\beta H})^2]}{\tr[ e^{-2\beta H}]}}{2^{n}(2^n+1)}
	\le
	\lVert T\rVert^2(\tr[e^{-\beta H}])^2\frac{\tr[\sigma_\beta^2]}{2^{n}(2^n+1)}
    \nonumber\\
    &\le
	\lVert T\rVert^2(\tr[e^{-\beta H}])^2\frac{\tr[\sigma_\beta^2]}{2^{2n}}.
\end{align}
From these we obtain two Chebyshev inequalities,
\begin{align}
\begin{split}
	&\text{Pr}\left[\Big|\frac{g-\bE_U[g]}{\bE_U[g]}\Big| > \delta \right] 
	< \frac{\mathrm{Var}_U[g]}{\delta^2\bE_U[g]^2}
	\le \frac{\tr[\sigma_\beta^2]}{\delta^2},
	\\
	&\text{Pr}\left[\Big|\frac{f-\bE_U[f]}{\bE_U[g]}\Big| > \delta \right] 
	< \frac{\mathrm{Var}_U[f]}{\delta^2\bE_U[g]^2}
	\le \lVert T\rVert^2\frac{\tr[\sigma_\beta^2]}{\delta^2}.
\end{split}
\end{align}
Accordingly, we define the sets of unitary operators,
\begin{align}
	G := \left\{U \left|\, \Big|\frac{g-\bE[g]}{\bE[g]}\Big| < \delta \right.\right\},
	\qquad
	F := \left\{U \left|\, \Big|\frac{f-\bE[f]}{\bE[g]}\Big| < \delta \right.\right\}, 
\end{align}
and their respective compliments $\bar{G}$ and $\bar{F}$.

We first bound the bias in the mean,
\begin{equation}
\label{eq:bias}
\begin{aligned}
    &\bE_U\left[\frac{f}{g} - \frac{\bE_U[f]}{\bE_U[g]}\right]
    \\
    &=
    \text{Pr}[F\cap G]\bE_{U|_{F\cap G}}
    \left[\frac{f}{g} - \frac{\bE_U[f]}{\bE_U[g]}\right]
    + \text{Pr}[\bar{F}\cup\bar{G}]\bE_{U|_{\bar{F}\cup\bar{G}}}
    \left[\frac{f}{g} - \frac{\bE_U[f]}{\bE_U[g]}\right]
    \\
    &\le
    \bE_{U|_{F\cap G}}\left[\frac{f}{g} - \frac{\bE_U[f]}{\bE_U[g]}\right]
    + 
    \frac{(1+\lVert T\rVert^2)\tr[\sigma_\beta^2]}{\delta^2}2\lVert T\rVert,
\end{aligned}
\end{equation}
where $\bE_{U|_{F\cap G}}[\cdot]$ stands for the average of $U$ over the intersection between $k$-design and $F\cap G$, and $\bE_{U|_{\bar{F}\cup \bar{G}}}[\cdot]$ is defined similarly.
The first term is converted to
\begin{equation}
\label{eq:bias_first}
\begin{aligned}
    &\bE_{U|_{F\cap G}}\left[\frac{f}{g} - \frac{\bE_U[f]}{\bE_U[g]}\right]
    \\
    &=
    \bE_{U|_{F\cap G}}\left[\frac{\bE_U[f]}{\bE_U[g]}\left(1+\frac{f-\bE_U[f]}{\bE_U[f]}\right)
    \left(1+R\left(\frac{g-\bE_U[g]}{\bE_U[g]}\right)\right)
    - \frac{\bE_U[f]}{\bE_U[g]}\right].
\end{aligned}
\end{equation}
where we defined $R(x):= (1+x)^{-1}-1$.
We note that the magnitude of $R(x)$ for $x\in[-\delta,\delta]\subset[-1,1]$ is bounded as
\begin{align}
\label{eq:remainder}
    \left|R(x)\right| \le \frac{\delta}{(1-\delta)^2}.
\end{align}
As specified later, $\delta$ is taken to satisfy $0<\delta\le1/2$.
Therefore, the magnitude of \eqref{eq:bias_first} is bounded as
\begin{equation}
\begin{aligned}
    &\bE_{U|_{F\cap G}}\left[\frac{f}{g} - \frac{\bE_U[f]}{\bE_U[g]}\right]
    \le
    \delta + (\lVert T\rVert + \delta)\frac{\delta}{(1-\delta)^2}.
\end{aligned}
\end{equation}
Plugging this into \eqref{eq:bias}, we find
\begin{align}
\begin{split}
	\left|\bE_U\left[\frac{f}{g} - \frac{\bE_U[f]}{\bE_U[g]}\right]\right|
	&\le 
	\delta +\frac{\delta}{(1-\delta)^2}\left( \lVert T\rVert + \delta\right)
	+
	\frac{(1+\lVert T\rVert^2)\tr[\sigma^2]}{\delta^2}2\lVert T\rVert.
\end{split}
\end{align}
Setting $\delta=(\tr[\sigma^2])^{1/3}/2$, we have
\begin{align}
\label{eq:tpqexperrbound}
\begin{split}
	\left|\bE_U\left[\frac{f}{g} - \frac{\bE_U[f]}{\bE_U[g]}\right]\right|
	&\le 
	(\tr[\sigma^2])^{1/3}
    \left(
    \frac{3}{2} 
    + 10\lVert T\rVert 
    + 8\lVert T\rVert^3
    \right),
\end{split}
\end{align}
where we used $0<\delta\le1/2$.

We turn to the variance.
Using the same technique, we proceed as follows,
\begin{align}
\label{eq:variance}
\begin{split}
    &\mathrm{Var}_U\left[\frac{f}{g}\right]
    \\
    &=
    \bE_U\left[\left(\frac{f}{g} - \bE_U\left[\frac{f}{g}\right]\right)^2 \right]
    \\
    &=
    \mathrm{Pr}[F\cap G]\bE_{U|_{F\cap G}}\left[
    \frac{f^2}{g^2}
    \right]
    + \mathrm{Pr}[\bar{F}\cup \bar{G}]\bE_{U|_{\bar{F}\cup \bar{G}}}\left[
    \frac{f^2}{g^2}
    \right]
    \\
    &\quad -
    \left(\mathrm{Pr}[F\cap G]\bE_{U|_{F\cap G}}\left[
    \frac{f}{g}
    \right]
    + \mathrm{Pr}[\bar{F}\cup \bar{G}]\bE_{U|_{\bar{F}\cup \bar{G}}}\left[
    \frac{f}{g}
    \right]\right)^2
    \\
    &\le
    \mathrm{Pr}[F\cap G]\bE_{U|_{F\cap G}}\left[
    \frac{f^2}{g^2}
    \right]
    -
    \left(\mathrm{Pr}[F\cap G]\bE_{U|_{F\cap G}}\left[
    \frac{f}{g}
    \right]
    \right)^2
    + 
    \mathrm{Pr}[\bar{F}\cup \bar{G}]\cdot 4\lVert T\rVert^2
    \\
    &\le
    \left(
    \bE_{U|_{F\cap G}}\left[\frac{f^2}{g^2}\right]
    - \bE_{U|_{F\cap G}}\left[\frac{f}{g}\right]^2
    \right)
    +
    \mathrm{Pr}[\bar{F}\cup \bar{G}]\cdot 5\lVert T\rVert^2.
\end{split}
\end{align} 
The first term is expanded as,
\begin{align}
\label{eq:first_var}
    &\bE_{U|_{F\cap G}}\left[\frac{f^2}{g^2}\right]
    - \bE_{U|_{F\cap G}}\left[\frac{f}{g}\right]^2
    \nonumber\\
    &=
    \frac{\bE[f]^2}{\bE[g]^2}\bE_{U|_{F\cap G}}\left[
    \left(1 + \frac{f-\bE[f]}{\bE[f]}
    + R + \frac{f-\bE[f]}{\bE[f]}R\right)^2
    \right]
    \nonumber\\
    &\quad - 
    \frac{\bE[f]^2}{\bE[g]^2}\bE_{U|_{F\cap G}}\left[
    1 + \frac{f-\bE[f]}{\bE[f]}
    + R + \frac{f-\bE[f]}{\bE[f]}R
    \right]^2
    \nonumber\\
    &=
    \frac{\bE[f]^2}{\bE[g]^2}\bE_{U|_{F\cap G}}\left[
    \left(\frac{f-\bE[f]}{\bE[f]}\right)^2
    + R^2 + 2\frac{f-\bE[f]}{\bE[f]}R 
    \right]
    \nonumber\\
    &\quad - 
    \frac{\bE[f]^2}{\bE[g]^2}\left(
    \bE_{U|_{F\cap G}}\left[\frac{f-\bE[f]}{\bE[f]}\right]^2
    + \bE_{U|_{F\cap G}}\left[R\right]^2
    +2\bE_{U|_{F\cap G}}\left[\frac{f-\bE[f]}{\bE[f]}\right]
    \bE_{U|_{F\cap G}}\left[R\right]
    \right)
    \nonumber\\
    &\quad + \calO(\delta^3).
\end{align}
Using the inequality~\eqref{eq:remainder}, Eq.~\eqref{eq:first_var} is bounded as
\begin{align}
\begin{split}
    &\bE_{F\cap G}\left[\frac{f^2}{g^2}\right]
    - \bE_{F\cap G}\left[\frac{f}{g}\right]^2
    \le
    2\left(
    \delta^2
    +
    \frac{\delta^2}{(1-\delta)^4}\lVert T\rVert^2
    +
    2 \frac{\delta^2}{(1-\delta)^2}\lVert T\rVert
    \right)
    + \calO(\delta^3).
\end{split}
\end{align}
This leads to an upper bound of variance~\eqref{eq:variance},
\begin{align}
\begin{split}
    \mathrm{Var}_U\left[\frac{f}{g}\right]
    &\le
    2\delta^2\left(
    1 + \frac{\lVert T\rVert}{(1-\delta)^2}
    \right)^2
    +
    \frac{(1+\lVert T\rVert^2)\tr[\sigma_\beta^2]}{\delta^2}5\lVert T\rVert^2
    + \calO(\delta^3).
\end{split}
\end{align}
Setting $\delta=(\tr[\sigma_\beta^2])^{1/4}/2$, we have
\begin{align}
\label{eq:variance_bound}
\begin{split}
    \mathrm{Var}_U\left[\frac{f}{g}\right]
    &\le
    (\tr[\sigma_\beta^2])^{1/2}
    \left(
    \frac{1}{2}
    + 4\lVert T\rVert
    + 28\lVert T\rVert^2
    + 20\lVert T\rVert^4
    \right)
    + \calO((\tr[\sigma_\beta^2])^{3/4}).
\end{split}
\end{align}

Substituting Eqs.~\eqref{eq:bias} and \eqref{eq:variance} into Eq.~\eqref{eq:TPQ_mse}, we find
\begin{equation}
\label{eq:TPQ_mse_bound}
\begin{aligned}
    &\mathbb{E}_U[(\bra{\psi_\beta}T_j\ket{\psi_\beta}-\tr[\sigma_\beta T_j])^2]
    \\
    &\le
    (\tr[\sigma_\beta^2])^{1/2}
    \left(
    \frac{1}{2}
    + 4\lVert T\rVert
    + 28\lVert T\rVert^2
    + 20\lVert T\rVert^4
    \right)
    + \calO((\tr[\sigma_\beta^2])^{2/3}).
\end{aligned}
\end{equation}
     
Thus, $\ket{\psi_\beta}$ is a TPQ state if the purity vanishes exponentially with system size, $\tr[\sigma_\beta^2]= \calO(2^{-\alpha n})$. Setting $H= H^\tau$, $\beta= \beta^\tau$, and substituting the SDP constraint matrices, $\lVert T\rVert= \lVert A_j\rVert\leq 1$ in Eqs.~\eqref{eq:tpqexperrbound} and \eqref{eq:TPQ_mse_bound}, we arrive at Eqs.~\eqref{eq:TPQ_mse_main} and~\eqref{eq:sdptpq} in the main text.

\section{A Spectral Condition for Exponentially Vanishing Purity}\label{app:exponential decaying purity}
In this appendix, we formulate a condition on the eigenvalue spectrum of the Hamiltonian $H$ which ensures an exponentially vanishing purity. Therefore, if the SDP Hamiltonian $H^\tau$ satisfies this condition, $\ket{\psi_\tau}$ in Eq.~\eqref{eq:TPQ} is indeed a valid TPQ state (Definition~\ref{def:TPQ_state}) and can be used to solve SDPs. We start by giving a derivation of the condition in Section~\ref{sec:derivspec}. In Section~\ref{sec:app_cond_hams} we show some analytical and numerical examples of Hamiltonians that satisfy the condition. In Section~\ref{sec:free-energy} we also provide an alternative argument for a vanishing purity based on the free energy. 

\subsection{Derivation of the Spectral Condition}
\label{sec:derivspec}
We first upper bound the purity of the Gibbs state by
\begin{equation}
\begin{aligned}
\label{eq:21}
    \tr[\sigma_\beta^2] = 
    \tr\left[\frac{e^{- 2\beta H}}{\left(\tr[e^{- \beta H}]\right)^2}\right] \leq
    \frac{2^n e^{-2\beta\lambda_\mathrm{min}}}{(\tr [e^{-\beta H}])^2}
    =
    \frac{2^{n}}{\left(\tr\left[ e^{-\beta(H-\lambda_\mathrm{min})}\right]\right)^2},
\end{aligned}
\end{equation} where $\lambda_{\mathrm{min}}$ is the smallest eigenvalue of $H$. 

Now let $\Pi_\nu$ be a projector onto the eigenspace of $H/n$ with eigenvalues less than $\frac{\lambda_\mathrm{min}}{n}+\nu$ for some positive constant $\nu$. Then we can write
\begin{equation}
\begin{aligned}
\label{eq:22}
    \big(\tr[ e^{-\beta(H-\lambda_\mathrm{min})}]\big)^2
    =
    \Big(\tr[ e^{-\beta n\big(\frac{H}{n}-\frac{\lambda_\mathrm{min}}{n}\big)}]\Big)^2
    \geq
    \big(\tr[\Pi_\nu e^{-\beta n\nu}]\big)^2 = e^{-2\beta n\nu} (\tr[\Pi_\nu])^2. 
\end{aligned}
\end{equation} 

By performing the trace in the eigenbasis of $H$ we see that $\tr[\Pi_\nu]$ counts the number of eigenvalues of $H$ that are within $n\nu$ of the smallest eigenvalue $\lambda_{\mathrm{min}}$. Combining Eq.~\eqref{eq:21} with Eq.~\eqref{eq:22} we thus arrive at 
\begin{equation}
\label{eq:purity}
    \tr[{\sigma_\beta}^2] \leq \frac{2^{n+\frac{2n\beta\nu}{\ln{2}}}}{(\tr[\Pi_\nu])^2}. 
\end{equation}

In order to make the right hand side vanish exponentially with $n$, we require that the $\tr[\Pi_\nu]$ increases exponentially with $n$ (the denominator should be bigger than the numerator for all $n$). Therefore, by introducing the condition $\tr[\Pi_\nu] \geq c 2^n$ for some constant $c$ independent of $n$ we find 
\begin{equation}
    \tr [\sigma_\beta^2]\leq\frac{2^{-\left(1-\frac{2\beta\nu}{\ln2}\right)n}}{c^2}, 
\end{equation} which decreases exponentially for $0\leq\nu < \ln{2}/(2\beta)$. In other words, we have shown that if the number of eigenvalues of $H$ in the interval $[\lambda_{\mathrm{min}}, \lambda_{\mathrm{min}} + \nu n ]$ is a constant fraction of the total number of eigenvalues, i.e. is larger than $c2^n$, the purity of the Gibbs state $\sigma$ vanishes with $n$.\footnote{Note that an equivalent condition on the rescaled Hamiltonian is that $H/n$ has $c2^n$ eigenvalues in the range $[\lambda_{\mathrm{min}}/n, \lambda_{\mathrm{min}}/n +\nu]$} This proves the following proposition. 

\begin{proposition}[a spectral condition for vanishing purity]
\label{prop:gen_pur_cond}
Given a Hermitian matrix $H$ of size $2^n\times2^n$ which has $c2^n\leq 2^n$ eigenvalues in the range $[\lambda_{\mathrm{min}}, \lambda_{\mathrm{min}} + \nu n ]$ for some constants $c$, $\beta$ and $0\leq\nu\leq \ln{2}/(2\beta)$. The purity of the Gibbs state $\sigma = e^{-\beta H}/\tr[e^{-\beta H}]$ can be upper bounded by $\tr[\sigma_\beta^2]\leq\frac{2^{-\left(1-\frac{2\beta\nu}{\ln2}\right)n}}{c^2}$.
\end{proposition}

The condition holds for generic Hamiltonians, hence generic Gibbs states. In order to apply it to the SDP Hamiltonian $H^\tau$ and arrive at the (more specialized) Proposition~\ref{prop:sdp_pur_cond} in the main text we take a few more steps. First, note that the condition needs to be satisfied for all update steps $\tau\leq T=8n/\epsilon^2 $ of the MMW method. As the purity is an increasing function with respect to $\beta$ we therefore focus on largest SDP re-scaling parameter $\beta^T=2n /\epsilon$. Importantly, $\beta^T$ is proportional to $n$, whereas the condition assumes a constant $\beta$. For this reason, we set $\beta_{\mathrm{max}}=2/\epsilon$ and $H=nH^T$ to make the connection \begin{equation}
    \rho_T = \frac{e^{-\beta_{\mathrm{max}} H^T}}{\tr[e^{-\beta_{\mathrm{max}}H^T}]} = \frac{e^{-\beta H}}{\tr[e^{-\beta H}]}.
\end{equation} Now we apply the generic Proposition \ref{prop:gen_pur_cond} and see that $H=nH^T$ needs to have $c2^n\leq 2^n$ eigenvalues in the range $[\lambda_{\mathrm{min}}, \lambda_{\mathrm{min}} + \nu n ]$. In other words, $H^T$ needs to have $c2^n\leq 2^n$ eigenvalues in the range $[\lambda_{\mathrm{min}}/n, \lambda_{\mathrm{min}}/n + \nu ]$ for $0\leq\nu\leq \epsilon\ln{2}/4$. Using $\tr[\rho_\tau^2]\leq\tr[\rho_T^2]$  then leads to Proposition~\ref{prop:sdp_pur_cond} in the main text. 

\subsection{Example classes of Hamiltonians that satisfy the spectral condition}
\label{sec:app_cond_hams}

Here we investigate the purity of the Gibbs state $\sigma_\beta$ for various example Hamiltonians. First we show that the purity of the Gibbs state of random Hamiltonian from the Generalized Unitary Ensemble (GUE) vanishes exponentially with $n$. Afterward, we verify numerically that the purity of the XXZ Heisenberg model and the two-dimensional spinless Hubbard model also vanishes exponentially with $n$. 

\subsubsection{GUE Random Matrices and Wigner's Semicircle law}

We first write the Gibbs state as \begin{equation}
    \sigma_\beta = \frac{e^{-\beta H}}{\tr[e^{-\beta H}]} = \frac{e^{-n\tilde{\beta} h}}{\tr[e^{-n\tilde{\beta} h}]}, 
\end{equation} where we have defined $\tilde{\beta}$ such that the rescaled Hamiltonian $h$ has eigenvalues between $-2$ and $2$. We assume $\tilde{\beta}$ is independent of $n$, which can be done in the SDP setting. We now assume that the rescaled Hamiltonian $h$ is a random GUE matrix and attempt to show that the purity vanishes with system size.  

Random matrices in the Gaussian Unitary Ensemble (GUE) are defined as follows. 
\begin{definition}[GUE random matrix, Ref.~\cite{anderson2010}]
    The $N\times N$ Gaussian unitary Ensemble is a family of complex Hermitian matrices specified by \begin{equation}
        h_{ij} = \frac{g_{ij}+ig_{ij}}{\sqrt{2N}} \quad \mathrm{if} j>i, 
    \end{equation}
    \begin{equation}
        h_{ii} = \frac{g_{ii}}{\sqrt{N}} \quad \mathrm{if} i=j,
    \end{equation} where $g_{ii}, g_{ij}, g'_{ij}$ for all $i,j$ are independent standard Gaussian random variables with mean zero and standard deviation $1$. 
\end{definition}
The empirical spectral density of a random matrix $h^{\mathrm{GUE}}$ is given by 
\begin{equation}
    \rho(E) := \frac{1}{N}\sum_{i=1}^N\delta[E-\lambda_i(h^{\mathrm{GUE}})],
\end{equation} where $\delta(\cdot)$ is the Dirac delta function and $\lambda_i$ the $i-$th eigenvalue of $h^{\mathrm{GUE}}$. Note that due to the randomness this empirical spectral density changes for every random instance.

The number of eigenvalues in the interval $[\lambda_0, \lambda_0+\nu]$ of the rescaled Hamiltonian $h$ is therefore given by
\begin{equation}
    \tr[\Pi_\nu] = N\int_{\lambda_0}^{\lambda_0+\nu} \rho(E)\mathrm{d}E.
\end{equation} In order to bound this integral we make use of the following fact. 
\begin{lemma}[spectral density bound, Refs.~\cite{erdos2017,Chen2023}]
For a random instance $h^{\mathrm{GUE}}$ of the GUE ensemble with dimension $N$, with empirical spectral density $\rho(E)$, there is an absolute constant $d$ such that we have 
\begin{equation}\label{eq:fact}
    \underset{E}{\mathrm{sup}} 
    \Big|\int_{-\infty}^E[\rho(E')-\rho_{\mathrm{sc}}(E')]\mathrm{d}E'\Big|
    \leq \frac{d}{\sqrt{N}}
\end{equation} with probability at least $1-\frac{1}{N}$. Here $\rho_{\mathrm{sc}}(E)$ is the spectral density of the Wigner's semi-circle law given by 
\begin{equation}
    \rho_{\mathrm{sc}}(E) := \frac{\sqrt{4-E^2}}{2\pi}.
\end{equation} 
\end{lemma}

Using $ -2\leq \lambda_0 \leq 2$ and integrating the semi-circle gives
\begin{align}
    \int_{-\infty}^{\lambda_0+\nu} \rho_{\mathrm{sc}}(E)\mathrm{d}E &= \int_{-2}^{\lambda_0+\nu} \frac{\sqrt{4-E^2}}{2\pi} \mathrm{d}E \\
    &= \frac{1}{4\pi}\sqrt{4-(\lambda_0+\nu)^2}(\lambda_0+\nu) + \sin^{-1}[\frac{\lambda_0+\nu}{2}] +\frac{1}{2} 
\end{align}

Thus, combining with Eq.~\eqref{eq:fact} we find
\begin{align}
    \tr[\Pi_\nu] &=  N\int_{-\infty}^{\lambda_0+\nu} \rho(E)\mathrm{d}E \\ &\geq N \int_{-\infty}^{\lambda_0+\nu} \rho_{\mathrm{sc}}(E)\mathrm{d}E - \sqrt{N}d \\
    &\geq \frac{N}{4\pi}\sqrt{4-(\lambda_0+\nu)^2}(\lambda_0+\nu) + \frac{N}{\pi}\sin^{-1}[\frac{\lambda_0\nu}{2}] +\frac{N}{2} - \sqrt{N}d \\
    &=: (\alpha N -\sqrt{N}d)
\end{align} with probability $1-\frac{1}{N}$. Here we introduced the constant $\alpha$, and we choose $\nu$ such that $-2<\lambda_0+\nu<2$ and $0\leq\nu\leq \frac{\ln 2}{2\tilde{\beta}}$. Most importantly $\nu$ is independent of $n$.  

Plugging this into Eq.~\eqref{eq:purity} we see that with probability $1-2^{-n}$, 
\begin{equation}
    \tr[{\sigma_{\tilde{\beta}}}^2] \leq \frac{2^{n+\frac{2n\tilde{\beta}\nu}{\ln{2}}}}{(\alpha 2^n - d 2^{n/2})^2} = \frac{2^{n+\frac{2n\tilde{\beta}\nu}{\ln{2}}}}{\alpha^2 2^{2n} - 2\alpha d 2^{3n/2} + d^2 2^{n}} = \frac{2^{-\left(1-\frac{2\tilde{\beta}\nu}{\ln2}\right)n}}{\alpha^2 -2\alpha d 2^{-n/2} + d^2 2^{-n}},
\end{equation} which decreases exponentially with system size $n$ for $0\leq\nu\leq \frac{\ln 2}{2\tilde{\beta}}$ and $\tilde{\beta}$ independent of $n$.

\subsubsection{Numerical Verification for Geometrically Local Hamiltonians}

In the main text, we demonstrate our TPQ-SDP solver to the Hamiltonian learning problem of the XXZ Heisenberg model [Eq.~\eqref{eq:xxz}] and the two-dimensional spinless Hubbard model [Eq.~\eqref{eq:spinlesshubbard}]. In order to guarantee the success of our algorithm for solving these type of SDPs, we verify that the purities of the Gibbs states of these Hamiltonians vanish exponentially with system size. In the top row of Fig.~\ref{fig:verifpur}, we show the purity of these Hamiltonians for system sizes of up to 13 qubits. We clearly observe an exponential decay, with the rate determined by the inverse temperature $\beta$. In the bottom row of the same figure, we also show the number of eigenvalues in the range $[\lambda_{\mathrm{min}}, \lambda_{\mathrm{min}}+\nu n]$ with $\nu= (1-10^{-5})\ln{2}/2\beta$. We see that the number of eigenvalues increases exponentially with the number of qubits, which is in accordance with our spectral condition (Proposition~\ref{prop:gen_pur_cond}). 

\begin{figure}[h!]
\centering
 \includegraphics[width=\linewidth]{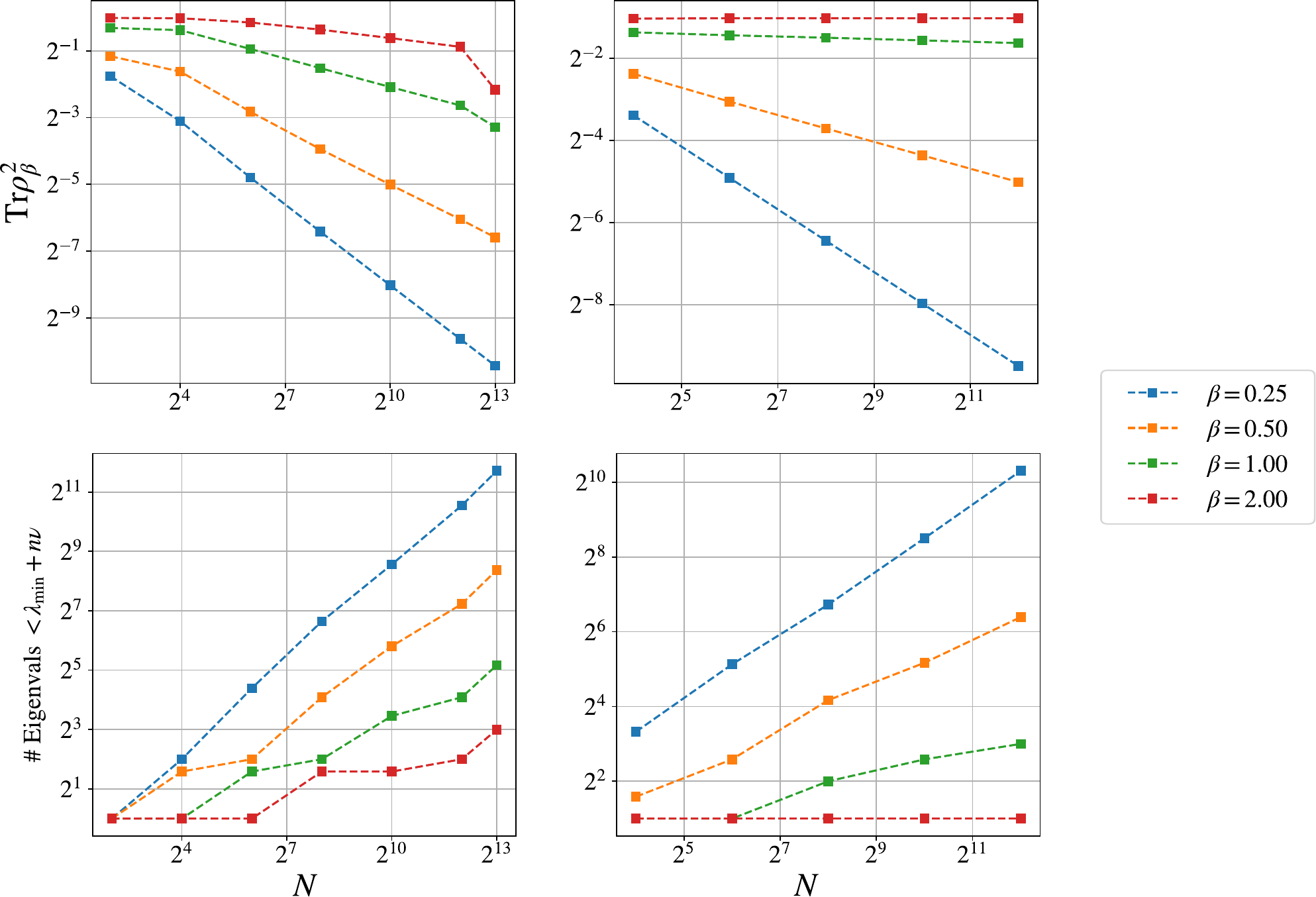}
 \caption{Scaling of the purity (top row) and number of eigenvalues (bottom row) within a constant $\nu$ of the re-scaled ground state energy $\lambda_{\mathrm{min}}$ for the XXZ Heisenberg model (left column) and the spinless two-dimensional Hubbard model (right column).} 
\label{fig:verifpur}
 \end{figure}

\subsection{Free energy argument for exponentially vanishing purity.}
\label{sec:free-energy}

As an alternative to the arguments above, here we present another argument for a vanishing purity by writing it in the terms of the free energy. The free energy of the Gibbs state is defined to be 
\begin{equation}
    F_\beta = -\frac{\ln{\tr [e^{-\beta H}]}}{\beta},
\end{equation} for some constant inverse temperature $\beta$. Then the purity can be written as
\begin{equation}
   \tr[\sigma_\beta^2] =\frac{\tr[e^{-2\beta H}]}{(\tr[e^{- H}])^2} = e^{-2\beta(F_{2\beta } -F_\beta)}.
\end{equation}
Since $F_{2\beta}>F_\beta$, the right hand side vanishes if $F_\beta \propto n$. In other words, the purity vanishes exponentially with system size if the free energy of the Hamiltonian is an extensive quantity (grows with the system size). As argued in \cite{Sugiura2013}, many physical (thermodynamic) systems have extensive free energy. Hence, we expect this argument to be applicable to a large class of Hamiltonians. 

\section{Proof of Theorem~\ref{thm:tpq_complexity}}
\label{app:tpq_complexity}

In this appendix, we prove the gate complexity of the approximate circuit implementation of the TPQ states used in the main text. We first upper bound the error between an expectation value of the exact TPQ state, $\ket{\psi_\tau}$, and the approximate TPQ state $\ket{\tilde{\psi}_\tau}$. Then we bound the success probability of preparing $\ket{\tilde{\psi}_\tau}$. By combining these results we arrive at Theorem~\ref{thm:tpq_complexity} in the main text.

\subsection{Bounding the error in the expectation values}
The circuit implementation of the TPQ state is given by 
\begin{equation}
    \ket{\tilde{\psi}_\tau} 
    =
    (\bra{0^r}\otimes I)
    \big( U_{P^\mathrm{exp}_{\beta^\tau}(K^\tau)/2} U \ket{0^n}\big)
    (\ket{0^r}\otimes I) 
    =
    \frac{ P^\mathrm{exp}_{\beta^\tau/2}(K^\tau) U\ket{0}}{\sqrt{\bra{0}U^\dag P^\mathrm{exp}_{\beta^\tau/2}(K^\tau)^2 U\ket{0}}},
\end{equation}
where $U_{P^\mathrm{exp}_{\beta^\tau}(K^\tau)/2}$ is the block encoding of the matrix exponential obtained with QET (Lemma~\ref{lem:qet_matrixexp}) acting on the $n$-qubit system and $r$-qubit ancillary registers. $U$ is a unitary $k\geq2$-design, e.g. a random Clifford circuit ($k=3$).

We now derive an upper bound on the QET polynomial approximation error,
\begin{equation}
\label{eq:exp_diff}
\begin{aligned}
    &\big|\bra{\psi_{\tau}} T_j \ket{\psi_{\tau}} 
    - \bra{\tilde{\psi}_{\tau}} T_j \ket{\tilde{\psi}_{\tau}}\big|
    \\
    &=\left| 
    \frac{\bra{0}U^\dag 
    e^{-\frac{\beta^\tau}{2} (K^\tau+I)}  T_j e^{-\frac{\beta^\tau}{2} (K^\tau+I)}U\ket{0}}{\bra{0}U^\dag e^{-\beta^\tau (K^\tau+I)} U\ket{0}} 
    - 
    \frac{\bra{0}P^\mathrm{exp}_{\beta^\tau/2}(K^\tau) T_j P^\mathrm{exp}_{\beta^\tau/2}(K^\tau) U\ket{0}}{\bra{0}U^\dag P^\mathrm{exp}_{\beta^\tau/2}(K^\tau)^2 U\ket{0}} 
    \right|,
\end{aligned}
\end{equation} 
of the TPQ expectation value of some bounded operator $\lVert T_j\rVert\leq1$. Recall that we wish to make this error $\mathcal{O}(\xi)$ since the exact TPQ state expectation value is within $\xi$ from the corresponding Gibbs state expectation value. We can then determine the parameter $\mu$ in the degree of the polynomial in Lemma \ref{lem:qet_matrixexp} in the main text. 

Let us drop the subscript $\tau$ for the moment and define,
\begin{equation}
\label{eq:pdef}
\begin{array}{ll}
    q
    :=\bra{0}U^\dag e^{-\beta (K+I)/2}  T_j e^{-\beta (K+I)/2}U\ket{0},
    \quad
    &
    \tilde{q} 
    :=\bra{0}U^\dag P^\mathrm{exp}_{\beta/2}(K) T_j P^\mathrm{exp}_{\beta/2}(K) U\ket{0},
    \\
    p
    :=\bra{0}U^\dag e^{-\beta(K+I)}  U\ket{0},
    &
    \tilde{p}
    :=\bra{0}U^\dag P^\mathrm{exp}_{\beta/2}(K)^2 U\ket{0}.
\end{array}
\end{equation}
Then we bound 
\begin{align}
\label{eq:bound_numerator}
    &|q-\tilde{q}|
    \nonumber\\
    &\le 
    |\bra{0}U^\dag e^{-\frac{\beta}{2}(K+I)}T_j (e^{-\frac{\beta}{2}(K+I)}-P^\mathrm{exp}_{\beta/2})U\ket{0}|
    +
    |\bra{0}U^\dag (e^{-\frac{\beta}{2}(K+I)}-{P^\mathrm{exp}_{\beta/2}}^\dag)T_j P^\mathrm{exp}_{\beta/2} U\ket{0}|
    \nonumber\\
    &\le 
    \sqrt{\bra{0}U^\dag e^{-\frac{\beta}{2}(K+I)} U\ket{0}} 
    \sqrt{\bra{0}U^\dag
    (e^{-\frac{\beta}{2}(K+I)}-{P^\mathrm{exp}_{\beta/2}}^\dag)
    T_j^\dag T_j
    (e^{-\frac{\beta}{2}(K+I)}-P^\mathrm{exp}_{\beta/2})U\ket{0}}
    \nonumber\\
    & + 
    \sqrt{\bra{0}U^\dag {P^\mathrm{exp}_{\beta/2}}^\dag P^\mathrm{exp}_{\beta/2} U\ket{0}} 
    \sqrt{\bra{0}U^\dag(e^{-\frac{\beta}{2}(K+I)}-{P^\mathrm{exp}_{\beta/2}}^\dag)
    T_j^\dag T_j
    (e^{-\frac{\beta}{2}(K+I)}-P^\mathrm{exp}_{\beta/2})U\ket{0}}
    \nonumber\\
    &\overset{\eqref{eq:QET_error}}{\le} 
    \mu(\sqrt{p} +\sqrt{\tilde{p}})\lVert T_j\rVert
    \le 2\mu,
\end{align}
using the triangle and Cauchy-Schwarz inequalities and inserting Lemma~\ref{lem:qet_matrixexp}. 
Similarly, the difference between $p$ and $\tilde{p}$ is bounded as,
\begin{equation}
\label{eq:bound_denominator}
\begin{aligned}
    &|p-\tilde{p}|
    \le \mu(\sqrt{p} +\sqrt{\tilde{p}})
    \le 2\mu.
\end{aligned}
\end{equation}
With these inequalities, Eq.~\eqref{eq:exp_diff} is reduced to,
\begin{equation}
\label{eq:error_expectaion}
\begin{aligned}
    &\big|\bra{\psi_{\tau}} T_j \ket{\psi_{\tau}} 
    - \bra{\tilde{\psi}_{\tau}} T_j \ket{\tilde{\psi}_{\tau}}\big|
    =\left|\frac{q}{p}-\frac{\tilde{q}}{\tilde{p}}\right| 
    \leq \left|\frac{q\tilde{p}-\tilde{q} \tilde{p}}{\tilde{p}p}\right|
    +\left|\frac{\tilde{q}\tilde{p}-\tilde{q}p}{\tilde{p}p}\right|
    \\
    &\overset{\eqref{eq:bound_numerator},\eqref{eq:bound_denominator}}{\leq}
    \frac{2\mu(\tilde{q}+\tilde{p})}{p\tilde{p}}
    \leq\frac{4\mu}{p}.
\end{aligned}
\end{equation}
In the next subsection [see Eq.~\eqref{eq:p_probability_explicit}], we show that $p\leq \frac{2^{-n}e^{-1/2}}{2}$ with probability $4\tr[\rho_{\tau}^2]$. Hence, if we take $\mu=\frac{\xi}{8}\frac{2^{-n}e^{-1/2}}{2}$, then $\big|\bra{\psi_{\tau}} T_j \ket{\psi_{\tau}} - \bra{\tilde{\psi}_{\tau}} T_j \ket{\tilde{\psi}_{\tau}}\big|<\xi/2$ with probability at least $1-4\tr[\rho_{\tau}^2]$. Combining with the mean-squared error of the exact TPQ state $\ket{\psi_\tau}$ [Eq.~\eqref{eq:sdptpq}] in the main text, we can then probabilistically bound the QSP approximation error 
\begin{equation}
\label{eq:purethermalstate_tilde}
\begin{aligned}
    &\mathrm{Pr}\big[|\bra{\tilde{\psi}_\tau}T_j\ket{\tilde{\psi}_\tau} 
    - \tr[T_j\sigma_\tau]|\ge \xi\big]
    \\
    &\leq \mathrm{Pr}\big[|\bra{\tilde{\psi}_\tau}T_j\ket{\tilde{\psi}_\tau} 
    - \bra{\psi_\tau}T_j\ket{\psi_\tau}| 
    +|\bra{\psi_\tau}T_j\ket{\psi_\tau} 
    - \tr[T_j\rho_\tau]|
    \ge \xi\big]
    \\
    &\leq \mathrm{Pr}\big[|\bra{\tilde{\psi}_\tau}T_j\ket{\tilde{\psi}_\tau} 
    - \bra{\psi_\tau}T_j\ket{\psi_\tau}| 
    \ge \xi/2\big]+\mathrm{Pr}[|\bra{\psi_\tau}T_j\ket{\psi_\tau} 
    - \tr[T_j\sigma_\tau]|
    \ge \xi/2\big]
    \\
    &\overset{\eqref{eq:p_probability_explicit}, \eqref{eq:sdptpq}}{\leq} 4\tr[\rho_{\tau}^2]
    + \frac{\frac{105}{2}(\tr[\rho_{\tau}^2])^{1/2}
    + \calO\big[(\tr[\rho_{\tau}^2])^{2/3})\big]}{\xi^2}.
\end{aligned}
\end{equation} 
Thus, $\ket{\tilde{\psi}_\tau}$ is a TPQ state when the spectral condition (Prop.~\ref{prop:sdp_pur_cond}) is satisfied, i.e., when $\tr[\rho_\tau^2]=\mathcal{O}(e^{-\alpha n})$. 

\subsection{Lower bounding the success probability and circuit complexity}
\label{app:prob_lower_bound}

We first lower bound the success probability of the QET circuit implementation of the TPQ state. By combining this with amplitude amplification and Eq.~\eqref{eq:purethermalstate_tilde}, we prove Theorem~\ref{thm:tpq_complexity} and Corollary~\ref{cor:tpqcomplex} for the circuit complexity in the main text. 

The probability of successfully preparing $\ket{\tilde{\psi}_\tau}$ is given by
\begin{equation}
\begin{aligned}
    p_{\mathrm{exp}} 
    &=
    \frac{1}{4}\bra{0}U^\dag P^\mathrm{exp}_{\beta/2}(K^\tau)^2 U\ket{0} 
    =
    \frac{\tilde{p}}{4},
\end{aligned} 
\end{equation} with $\tilde{p}$ from Eq.~\eqref{eq:pdef}. Using  Eq.~\eqref{eq:bound_denominator}, we get $\tilde{p} \geq p-2\mu$.
Therefore, lower bounding $p_{\mathrm{exp}}$ is reduced to lower bounding $p$. Since $p$ depends on the random circuit $U$, we bound $p$ probabilistically using Chebyshev's inequality. 

Averaging $p$ over the unitary $k$-design gives 
\begin{equation}
     \mathbb{E}_U\big[p\big]=\mathbb{E}_U\big[\bra{0}U^\dagger e^{-\beta^\tau( H^\tau-\Xi^\tau I)} U\ket{0}\big] 
    = \frac{\tr[e^{-\beta^\tau( H^\tau-\Xi^\tau)}]}{2^n},
\end{equation}
with variance
\begin{equation}
\begin{aligned}
    \mathrm{Var}_U[p]
    &= \mathrm{Var}_U\big[\bra{0}U^\dagger e^{-\beta^\tau( H^\tau-\Xi^\tau)} U\ket{0}\big] 
    \\
    &= \frac{2^n\tr[e^{-2\beta^\tau( H^\tau-\Xi^\tau)}]
    -(\tr[e^{-\beta^\tau( H^\tau-\Xi^\tau)}])^2}{2^{2n}(2^n+1)}.
\end{aligned}
\end{equation}
By plugging these equations into Chebyshev's inequality we find
\begin{align}
\label{eq:pcheby}
\begin{split}
    &\mathrm{Pr}\big[p\leq \mathbb{E}_U[p](1-\delta)\big]
    \leq 
    \frac{\mathrm{Var}_U [p]}{\delta^2 (\mathbb{E}_U[p])^2} 
    \\
    &\Longleftrightarrow
    \mathrm{Pr}\left[p\leq \frac{\tr[e^{-\beta^\tau( H^\tau-\Xi^\tau)}]}{2^n} \left(1-\delta\right)\right]
    \leq 
    \frac{\tr[\rho_\tau^2]}{\delta^2}.
\end{split}
\end{align}
Using the fact that
\begin{equation}
\label{eq:succeseig}
    \tr[e^{-\beta^\tau( H^\tau-\Xi^\tau)}] = e^{-\beta^\tau(\lambda_{\mathrm{min}}^\tau-\Xi^\tau)}\tr[e^{-\beta^\tau( H^\tau-\lambda_{\mathrm{min}}^\tau)}] \geq e^{-1/2},
\end{equation} 
since $\lambda_{\mathrm{min}}^\tau-\Xi^\tau\leq\frac{1}{2\beta^\tau}$. 
Setting $\delta=1/2$, we arrive at
\begin{equation}
\label{eq:p_probability_explicit}
    \mathrm{Pr}\left[p\leq  2^{-n}\frac{e^{-1/2}}{2} \right]
    \leq 
    4\tr[\rho_\tau^2].
\end{equation}

Under our vanishing purity assumption, $\tr[\rho_\tau^2]=e^{-\alpha n}$, using $\mu=\frac{\xi}{8}\frac{2^{-n}e^{-1/2}}{2}$, the success probability, $p_{\mathrm{exp}}\geq p-2\mu$, is larger than $\mathcal{O}(2^{-n})$ with probability $1-4e^{-\alpha n}$. We can boost the success probability to $\mathcal{O}(1)$ with amplitude amplification using the number of rounds $\mathcal{O}(p_{\mathrm{exp}}^{-1/2})=\mathcal{O}(2^{n/2})=\mathcal{O}(N^{1/2})$, which succeeds with probability $1-\mathcal{O}(e^{-\alpha n})$~\cite{Martyn2021grand}.
We call the operator which implements $P^\mathrm{exp}_{\beta^\tau}(K^\tau)$ with $\calO(p_{\mathrm{exp}}^{-1/2})$ rounds of amplitude amplification $V_{\beta^\tau}$. 
The new (redefined) circuit implementation of the TPQ state is therefore given by 
\begin{equation}
    \ket{\tilde{\psi}_\tau} = V_{\beta^\tau}U\ket{0},
\end{equation} 
which prepares the TPQ state $\ket{\psi_\tau}$ to high enough accuracy with probability $\mathcal{O}(1)$. Combining this with $\beta^\tau \leq \frac{2 \log{N}}{\epsilon}$,  Lemma~\ref{lem:qet_matrixexp}, and 
$\mu=\frac{\xi 2^{-n}e^{-1/2}}{16}$, we arrive at
\begin{align}
    \calO(p_{\mathrm{exp}}^{-1/2}) 
    \times \calO\left(\sqrt{\beta^\tau}\log{\frac{1}{\mu}}\right)
    \times \mathcal{T}_{K^\tau} 
    = 
    \mathcal{O}\left(N^{1/2}\left(\frac{\log{N}}{\epsilon}\right)^{1/2}\log\left(\frac{N}{\xi}\right) 
    \cdot \mathcal{T}_{K^\tau}\right).
\end{align} 
for the circuit complexity of $V_{\beta^\tau}$.
Here $\mathcal{T}_{K^\tau}$ is the depth of the block encoding, which is assumed to be $\mathcal{O}(\log{N})$. This proves Theorem~\ref{thm:tpq_complexity} in the main text.

If we additionally assume the spectral condition~(Lemma \ref{prop:sdp_pur_cond}) is satisfied we can obtain a tighter bound. First, we write
\begin{equation}
     e^{-\beta^\tau(\lambda_{\mathrm{min}}^\tau-\Xi^\tau)}\tr[e^{-\beta^\tau( H^\tau-\lambda_{\mathrm{min}}^\tau)}] \geq e^{-\beta^\tau(\lambda_{\mathrm{min}}^\tau-\Xi^\tau)} e^{-\beta^\tau n \nu}\tr[\Pi_\nu],
\end{equation} 
where $\Pi_\nu$ is a projector on the eigenstates of $H^\tau$ with eigenvalues in the range $[\lambda_{\mathrm{min}} +n\nu]$ [compare Eq.~\eqref{eq:22}]. 
Using the spectral condition $\tr[\Pi_\nu]=c 2^n$ for $0\leq\nu\leq \ln{2}/(2\beta^\tau)$, and the fact that $\lambda_{\mathrm{min}}^\tau-\Xi^\tau\leq\frac{1}{2\beta^\tau}$, we get 
\begin{equation}
\label{eq:pbound}
    e^{-\beta^\tau(\lambda_{\mathrm{min}}^\tau-\Xi^\tau)}\tr[e^{-\beta^\tau( H^\tau-\lambda_{\mathrm{min}}^\tau)}] \geq ce^{-1/2}2^{-n/2}2^n, 
\end{equation} 
for $\nu = \ln 2/2\beta^\tau$. Substituting this into Eq.~\eqref{eq:pcheby} with $\delta=1/2$ gives 
\begin{equation}
\label{eq:p_probability_explicit_spectral}
    \mathrm{Pr}\left[p\leq \frac{c  2^{-n/2}}{2\sqrt{e}}\right]
    \leq
    4\tr[\rho_\tau^2].
\end{equation} 
Thus, $p_{\mathrm{exp}}$ is larger than $\mathcal{O}(2^{-n/2})$ with probability at least $1-\calO(e^{-\alpha n})$. The complexity of $V_{\beta^\tau}$ in the presence of the spectral condition is therefore $ \tilde{\mathcal{O}}\big(\frac{N^{1/4}}{\epsilon^{1/2}}\big)$, which proves Corollary~\ref{cor:tpqcomplex} in the main text. 

\subsection{Alternative implementation of the matrix exponential with QET}
\label{app:alternative_QETexp}

In the main text, we assumed we have access to a block encoding of the shifted Hamiltonian $K^\tau = H^\tau-(1+\Xi^\tau)I$, which has a spectrum in the range $[-1, 1]$. In practice, it is unclear how one can construct such a block encoding when given access to a block encoding for $H^\tau$. For example, an LCU block encoding requires sub-normalization which implies that the lowest eigenvalue of $K^\tau$ is not necessarily close to $-1$. This negatively affects the success probability~\cite{vanapeldoorn2019,vanApeldoorn2020quantumsdpsolvers}. Here we show an alternative approach for implementing $e^{-\beta^\tau H^\tau}$ with QET, which does not have this issue.  

Instead of using $K^\tau$, we introduce a rescaled Hamiltonian 
\begin{equation}
\label{eq:block_encoding_hamiltonian2}
    \tilde{H}^\tau = \left(1-\frac{1}{4\beta^T}\right)\frac{H^\tau-\Xi^\tau I }{2} + \frac{1}{4\beta^T}, 
\end{equation} 
which has all its eigenvalues in the interval $[1/(4\beta^T),1]$.
Recall that $\Xi$ is chosen so that $\lambda_\text{min}(H^\tau)-\Xi\le 1/(2\beta^T)$ as discussed in Section~\ref{subsec:prepTPQ}. $\tilde{H}^\tau$ can be block encoded with LCU without shifting the smallest eigenvalue. In order to construct a polynomial approximation to $e^{-\beta^\tau \tilde{H}^\tau}$ we use fixed parity QET. 

\begin{lemma}[Quantum eigenvalue transformation of fixed parity, restatement of Corollary 18 in Ref.~\cite{Gilyen2019}]\label{lem:QETfixedparity}
    Let $U_A$ be a block encoding of a Hermitian matrix $A$, which requires $a$ ancillary qubits. Let $P\in \mathbb{R}[x]$ be a real parity-$(d$ mod $2)$ polynomial of degree $d$ such that $\max_{x\in[-1,1]}|P(x)|\le 1$. Then there exists a block encoding $U^{\vec{\phi}}_\mathrm{QET}$ of $P(A)$ that uses $d$ queries to $U_A$ and $U_A^\dag$, a single application of controlled-$U_A$, $a+1$ ancillary qubits, and $\mathcal{O}((a+1)d)$ other elementary gates. Moreover, a set of $d+1$ angles $\vec{\phi}=\{\phi_0,\dots,\phi_{d}\}$ parameterising $U^{\vec{\phi}}_\mathrm{QET}$ can be computed classically efficiently.
\end{lemma}

Since $e^{-\beta x} = e^{-\beta |x|}\quad\forall x\in[0,1]$, and $e^{-\beta |x|}$ has even parity, we propose a polynomial approximation to $e^{-\beta |x|}$ which combined with Lemma~\ref{lem:QETfixedparity} provides a circuit implementation of an approximation to $e^{-\beta^\tau \tilde{H}^\tau}$. 
Let $P^{\mathrm{exp}}_{\beta/2}(x)$ be the polynomial approximation of $e^{-\frac{\beta}{2}(x+1)}$ from Lemma~\ref{lemma:exp_poly_approx} in the main text, and let  $P^\mathrm{sgn}_{\zeta,\Delta}$ be a polynomial approximation to the sign function defined in the following Lemma. 

\begin{lemma}[Polynomial approximation to sign function~\cite{Gilyen2019}]
\label{lemma:sign_poly_approx}
Let $0<\zeta<1$ and $0<\Delta<1$. There exists a polynomial $P^\mathrm{sgn}_{\zeta,\Delta}$ of parity odd and degree $\mathcal{O}(\frac{1}{\Delta}\log\frac{1}{\zeta})$, such that
\begin{enumerate}
    \item $|P^\mathrm{sgn}_{\zeta,\Delta}(x)|\leq 1$ for $x\in[-1,1]$.
    \item $|\mathrm{sgn}(x)-P^\mathrm{sgn}_{\zeta,\Delta}(x)|\leq\zeta$ for $x\in[-1,1]\backslash(-\frac{\Delta}{2},\frac{\Delta}{2})$.
\end{enumerate}
where $\mathrm{sgn}(x)$ is the sign function, i.e., $\mathrm{sgn}(x)=-1~(x<0)$, $0~(x=0)$, and $1~(x>0)$.
\end{lemma}

Then, the parity-even polynomial $Q_\beta (x) := P_{\beta/2}(2P^\mathrm{sgn}_{\zeta,\Delta}(x)x-1)$ satisfies
\begin{equation}
\begin{aligned}
    &\max_{x\in[\frac{\Delta}{2},1]} |Q_\beta (x)-e^{-\beta|x|}|
    \\
    &\leq \max_{x\in[\frac{\Delta}{2},1]} |Q_\beta (x)-e^{-\beta P^\mathrm{sgn}_{\zeta,\Delta}(x) x}|
    +\max_{x\in[\frac{\Delta}{2},1]} |e^{-\beta P^\mathrm{sgn}_{\zeta,\Delta}(x) x}-e^{-\beta|x|}|
    \\
    &\leq\mu+(1-e^{-\beta\zeta}),
\end{aligned}
\end{equation} 
where in the second inequality we have used Lemmas~\ref{lemma:exp_poly_approx} and \ref{lemma:sign_poly_approx} to bound the errors. 
Setting $\zeta = \frac{1}{\beta}\ln(\frac{1}{1-\mu})$,  we arrive at $\max_{x\in[\frac{\Delta}{2},1]} |Q_\beta (x)-e^{-\beta|x|}|\leq 2\mu$. 

Note that the interval on which this inequality holds is $[\Delta/2, 1]$. The smallest eigenvalue of $\tilde{H}^\tau$ must be smaller than $\mathcal{O}(1/2\beta^\tau)$ for a finite success probability [recall Eq.~\eqref{eq:succeseig}], which is indeed the case for the Hamiltonian given in Eq.~\eqref{eq:block_encoding_hamiltonian2}.
Since the spectral range of $\tilde{H}^\tau$ is in $[1/(4\beta^T),1]$, we set $\Delta=1/(4\beta^T)$.
The polynomial transformation of this Hamiltonian satisfies
\begin{equation}
\label{eq:poly_approx}
    \big\lVert Q_{\tilde{\beta}^\tau} (\tilde{H}^\tau)-e^{-\tilde{\beta}^\tau \tilde{H}^\tau} \big\rVert 
    \leq 2\mu,
\end{equation}
where $\tilde{\beta}^\tau = \frac{\beta^\tau}{1-1/(4\beta^T)}$, and it can be implemented with the query complexity at most
\begin{align}
\label{eq:poly_degree2}
    \mathcal{O}\left((\beta^\tau)^{3/2}\left(\log\frac{1}{\mu}\right)\log\left(\frac{\beta^\tau}{\mu}\right)\right).
\end{align}

Compared to the degree of the polynomial of $K^\tau$ in Lemma~\ref{lem:qet_matrixexp}, we observe that Eq.~\eqref{eq:poly_degree2} has a higher power for $\beta^\tau$. This means that the final TPQ state circuit complexity in Theorem~\ref{thm:tpq_complexity} in the main text would have  dependence $\epsilon^{-3/2}$. This is slightly worse but avoids the potential sub-normalization problem of the block encoding $K^\tau$. In addition, our construction is fully explicit so we can combine it with LCU to obtain the resource estimates in Table~\ref{tbl:resources}.

\section{OR Lemma and Broken Constraint Check}

\subsection{Proof of quantum OR lemma}
\label{app:orlemma}
The test stated in the quantum OR Lemma (Lemma~\ref{lemma:fastOR_QET}) can be implemented as follows~\cite{Gilyen2019}. 
For a projector $\Pi$ ($\Pi^2=\Pi$), the projector-controlled-NOT operator in the Lemma is given by 
\begin{equation}
    \mathrm{C}_{\Pi}\mathrm{NOT} 
    = \Pi \otimes X + (I-\Pi) \otimes I,
\end{equation} 
and $V:=\sum_{j=1}^{k}\mathrm{C}_{\Pi_j}\mathrm{NOT} \otimes \ket{j}\bra{j}$.
Then
\begin{align}
\label{eq:or_proj}
    (I\otimes \bra{0}\otimes\bra{0^{\log k}}H^{\otimes \log k})V(I\otimes\ket{0}\otimes H^{\otimes \log k}\ket{0^{\log k}})
    = \frac{1}{k}\sum_{j=1}^{k}(I-\Pi_j),
\end{align}
where $H$ is the Hadamard gate.
For $\lambda:=\frac{1-\upsilon}{2k}$ it was proved in~\cite{Harrow2017} that $\tr[\eta\Pi_{\le 1-\lambda}]\ge(1-\upsilon)^2/4$ in case $(i)$ and $\tr[\eta\Pi_{\le 1-4\lambda/5}] \le 5k\phi$ in case $(ii)$. Here, $\Pi_{\le y}$ for $y\in\mathbb{R}$ is a projector onto the space spanned by the eigenvectors of Eq.~\eqref{eq:or_proj} with eigenvalues less than $y$. The discrimination of eigenvalues larger than $1-\lambda$ or smaller than $1-4\lambda/5$ can be done with the eigenvalue discrimination (Lemma~\ref{lemma:eigenvalue_disc}).
Therefore, the eigenvalue discrimination with $b=1-4\lambda/5$, $a=1-\lambda$, and $\delta'=\zeta$ in Lemma~\ref{lemma:eigenvalue_disc} leads to the stated complexity in Lemma~\ref{lemma:fastOR_QET} in the main text.

\subsection{Proof of the eigenvalue Projection Lemma}
\label{app:eigendisclemma}
In this appendix, we prove Lemma~\ref{lemma:eigenvalue_disc_density} in the main text. 
To this end, we first introduce the following Lemma.
\begin{lemma}[Eigenvalue discrimination~\cite{Gilyen2019}]
\label{lemma:eigenvalue_disc}
    Let $b>a>0$, $0<\delta'<1$, and $U_\Lambda$ be a block encoding of a Hermitian matrix $\Lambda$. Let $\ket{\psi}$ be a quantum state given either by $\ket{\psi}=\sum_{\lambda\ge b}\alpha_\lambda\ket{\lambda}$ or $\ket{\psi}=\sum_{\lambda\le a}\beta_\lambda\ket{\lambda}$ for a set of eigenvalues and eigenstates $\{\lambda, \ket{\lambda}\}$ of $\Lambda$.
    Then, we can discriminate the two cases with error probability at most $\delta'$ with $\mathcal{O}((\max[b-a,\sqrt{1-a^2}-\sqrt{1-b^2}])^{-1}\log\frac{1}{\delta'})$ uses of $U_\Lambda$ and other elementary gates.
\end{lemma}
This Lemma allows us to construct an approximate eigenvalue discrimination operator $P^\text{ED}(\Lambda)$ using $\mathcal{O}((\max[b-a,\sqrt{1-a^2}-\sqrt{1-b^2}])^{-1}\log\frac{1}{\delta'})$ uses of $U_\Lambda$ and other gates.

Let us move on to the proof of Lemma~\ref{lemma:eigenvalue_disc_density}.
We repeat the Lemma for completeness. 
\begin{lemma}[Eigenvalue projection]
\label{lemma:eigenvalue_disc_density2}
    Let $b>a>0$, $0<\delta'<1$, and $U_\Lambda$ be a block encoding of a Hermitian matrix $\Lambda$. Suppose a quantum state~$\rho$ is promised to satisfy either $\tr[\rho \Pi_{\lambda\ge b}]\ge p_b$ or $\tr[\rho \Pi_{\lambda\le a}]\ge p_a$ for eigenvalues $\lambda$ of $\Lambda$.
    Then, we can construct an algorithm to accept $\rho$ with probability $p_b (1-\delta')^2$ in case $\tr[\rho \Pi_{\lambda\ge a}]\ge p_b$ holds, and accept it with probability $\delta'^2 + (1-p_a)$ in case $\tr[\rho \Pi_{\lambda\le a}]\ge p_a$ holds. The algorithm is implemented with $\mathcal{O}((\max[b-a,\sqrt{1-a^2}-\sqrt{1-b^2}])^{-1}\log\frac{1}{\delta'})$ uses of $U_\Lambda$ and other elementary gates.
\end{lemma}
\begin{proof}
    According to Lemma~\ref{lemma:eigenvalue_disc}, when $\tr[\rho \Pi_{\lambda\ge b}]\ge p_b$ is satisfied, we have
    \begin{equation}
        \tr[P^\text{ED}(\Lambda) \rho P^\text{ED}(\Lambda)]
        \ge
        p_b P^\text{ED}(\lambda\ge b)^2
        \ge p_b (1-\delta')^2.
    \end{equation}
    In case $\tr[\rho \Pi_{\lambda\le a}]\ge p_a$ is satisfied, we have
    \begin{equation}
    \begin{aligned}
        \tr[P^\text{ED}(\Lambda) \rho P^\text{ED}(\Lambda)]
        &\le
        P^\text{ED}(\lambda\le a)^2 + (1-p_a)P^\text{ED}(\lambda\ge b)^2
        \\
        &\le 
        \delta'^2 + (1-p_a).
    \end{aligned}
    \end{equation}
    This completes the proof.
\end{proof}
\subsection{Block-encoding of the Majority Vote Operator}
\label{app:majorblock}

Here, we show how to construct a block-encoding of the operator [Eq.~\eqref{eq:amp_proj}],
\begin{equation}
\label{eq:amp_projector}
    \Lambda:=\frac{1}{x}\sum_{i=0}^{x-1}\Lambda_i,
\end{equation} 
where $\Lambda_i=\ket{0^\ell}\bra{0^\ell}$. 
First, we note that a controlled $\Lambda_i$ reflection operator,
\begin{equation}
\label{eq:Lambda_reflection}
    \sum_{i=0}^{x-1} (2\Lambda_i - I)\otimes\ket{i}\bra{i},
\end{equation}
is implemented by LCU with $\tilde{\calO}(x)$ elementary gates and $y-1=\lceil{\log x}\rceil$ ancillary qubits.
Now, we append a single qubit to turn the reflection into the following operator,
\begin{equation}
    \sum_{i=0}^{x-1} (I\otimes\ket{0}\bra{0} + (2\Lambda_i - I)\otimes \ket{1}\bra{1})\otimes\ket{i}\bra{i}.
\end{equation}
Sandwiching it by $I\otimes H\otimes H^{\otimes \log x}$ and its conjugate, we find
\begin{align}
\label{eq:C_Lambda_NOT}
\begin{split}
    &\sum_{i=0}^{x-1} (I\otimes\ket{+}\bra{+} + (2\Lambda_i - I)\otimes\ket{-}\bra{-})\otimes H^{\otimes\log x}\ket{i}\bra{i}H^{\otimes\log x}
    \\
    &=\sum_{i=0}^{x-1} (\Lambda_i\otimes I + (I-\Lambda_i)\otimes X)\otimes H^{\otimes\log x}\ket{i}\bra{i}H^{\otimes\log x}
    \\
    &=:U_\Lambda.
\end{split}
\end{align}
Upon projecting onto the state $I\otimes \ket{0^{y}}$, we find
\begin{align}
    (I\otimes \bra{0^{y}})U_\Lambda(I\otimes \ket{0^{y}})
    =\frac{1}{x}\sum_{i=0}^{x-1}\Lambda_i
    \overset{\eqref{eq:amp_projector}}{=} \Lambda,
\end{align}
i.e., $U_\Lambda$ is a block encoding of $\Lambda$.

\subsection{Number of copies for amplifying probability gap}
\label{app:copies}

We wish to distinguish between case $(i)$ $p:=\tr[\rho_\tau \calA_j]\ge\frac{2+\epsilon}{4}+\epsilon_\text{gap}$ and case $(ii)$ $p:=\tr[\rho_\tau \calA_j]\le\frac{2+\epsilon}{4}$.
However, by measuring the ancillary registers of the $x$ copies we only have access to random variables $\{x_i\}_{i=1}^{x}$, each of which takes $1$ with probability
\begin{align}
    \tilde{p} 
    := \bra{\psi_\tau} \calA_j \ket{\psi_\tau},
\end{align}
and takes $0$ with probability $1-\tilde{p}$.
The two probabilities $p$ and $\tilde{p}$ are $\frac{\xi}{4}$-close with high probability,
\begin{align}
\label{eq:prob_TPQ}
    \text{Pr}\left[|p-\tilde{p}|\ge\frac{\xi}{4}\right]
    \le C_\xi e^{-\alpha n},
\end{align}
i.e., $\ket{\psi_\tau}$ is a TPQ state associated with the Gibbs state $\rho_\tau$ (Definition~\ref{def:TPQ_state}).

Using the $x$ random variables $\{x_i\}$ sampled from a fixed TPQ state $\ket{\psi_\tau}$, we conclude that the case $(i)$ is satisfied if $\sum_i\frac{x_i}{x} \ge \frac{2+\epsilon}{4}+\frac{\epsilon_\text{gap}}{2}$ and the case $(ii)$ is satisfied if $\sum_i\frac{x_i}{x} < \frac{2+\epsilon}{4}$.
We show that our conclusion is correct with high probability.
In the case $(i)$, the probability of acceptance is
\begin{align}
\begin{split}
    &\text{Pr}\left[
    \sum_i\frac{x_i}{x} 
    \ge \frac{2+\epsilon}{4}+\frac{\epsilon_\text{gap}}{2}
    \right]
    \\
    &\ge
    \text{Pr}\left[
    \sum_i\frac{x_i}{x} - \tilde{p}
    \ge - \left(\frac{\epsilon_\text{gap}}{2} - |p-\tilde{p}|\right)
    \right]
    \\
    &\ge
    \text{Pr}\left[
    \sum_i\frac{x_i}{x} - \tilde{p}
    \ge - \left(\frac{\epsilon_\text{gap}}{2} - |p-\tilde{p}|\right)
    \left|\,
    |p-\tilde{p}|\le \frac{\xi}{4}
    \right.
    \right]
    \text{Pr}\left[
    |p-\tilde{p}|\le \frac{\xi}{4}
    \right]
    \\
    &\ge
    1 
    - e^{-2x\big(\frac{\epsilon_\text{gap}}{2} - \frac{\xi}{4}\big)^2}
    - C_\xi e^{-\alpha n}.
\end{split}
\end{align}
We used $p\ge\frac{2+\epsilon}{4}+\epsilon_\text{gap}$ in the first inequality, and the last inequality follows from Chernoff bound.\footnote{
Chernoff bounds (one-sided Chernoff-Hoeffding inequalities) are given by
\begin{equation}
\label{eq:chernoff}
    \mathrm{Pr}\left(X-\mu \geq \delta\right)\le e^{-2\delta^2/n},
    \qquad
    \mathrm{Pr}\left(X-\mu \geq -\delta\right)\ge 1-e^{-2\delta^2/n}, 
\end{equation} 
where $X$ is the sum of $n$ independent random variables taking values in $[0, 1]$, $\mu$ is the sum's expectation value, and $\delta>0$.
}

Similarly, in the case $(ii)$, we find the probability of acceptance,
\begin{align}
\begin{split}
    &\text{Pr}\left[
    \sum_i\frac{x_i}{x} 
    \ge \frac{2+\epsilon}{4}+\frac{\epsilon_\text{gap}}{2}
    \right]
    \\
    &\le
    \text{Pr}\left[
    \sum_i\frac{x_i}{x} - \tilde{p}
    \ge \frac{\epsilon_\text{gap}}{2} - |p-\tilde{p}|
    \right]
    \\
    &=
    \text{Pr}\left[
    \sum_i\frac{x_i}{x} - \tilde{p}
    \ge - \left(\frac{\epsilon_\text{gap}}{2} - |p-\tilde{p}|\right)
    \left|\,
    |p-\tilde{p}|\le \frac{\xi}{4}
    \right.
    \right]
    \text{Pr}\left[
    |p-\tilde{p}|\le \frac{\xi}{4}
    \right]
    \\
    &+
    \text{Pr}\left[
    \sum_i\frac{x_i}{x} - \tilde{p}
    \ge - \left(\frac{\epsilon_\text{gap}}{2} - |p-\tilde{p}|\right)
    \left|\,
    |p-\tilde{p}|\ge \frac{\xi}{4}
    \right.
    \right]
    \text{Pr}\left[
    |p-\tilde{p}|\ge \frac{\xi}{4}
    \right]
    \\
    &\le
    e^{-2x\big(\frac{\epsilon_\text{gap}}{2} - \frac{\xi}{4}\big)^2}
    + C_\xi e^{-\alpha n}.
\end{split}
\end{align}
where we used $p\le\frac{2+\epsilon}{4}$ in the first inequality.

Setting $\epsilon_\text{gap}=\xi$, we correctly distinguish the two cases with probability at least
$1 - \frac{\delta}{m+1} - C_\xi e^{-\alpha n}$ using the number  of copies, $x=\frac{8}{\xi^2}\ln\frac{m+1}{\delta}$.

\subsection{Projectors}
\label{app:projector_inequality}

We derive Eqs.~\eqref{eq:projector_case1} and \eqref{eq:projector_case2} by calculating $\tr[(\tilde{\rho}\otimes\ket{0^y}\bra{0^y})\Pi_j]$.
Recall that $\tr[(\tilde{\rho}\otimes\ket{0^y}\bra{0^y})\Pi_j]$ is the probability that the projector $\Pi_j$ [Eq.~\eqref{eq:amp_proj}] accepts the state $\tilde{\rho}\otimes\ket{0^y}\bra{0^y}$ with
\begin{equation}
    \tilde{\rho}
    := 
    \bigotimes_{i=0}^{x-1} 
    \big(\ket{\tilde{\psi}_{\tau}}\bra{\tilde{\psi}_{\tau}}
    \otimes \ket{0^\ell}\bra{0^\ell}\big)
\end{equation}
a TPQ state $\ket{\tilde{\psi}}$.

Following the discussion in Appendix~\ref{app:copies}, we set $x=\frac{8}{\xi^2}\ln\frac{m+1}{\delta}$ in order to incorporate the TPQ error $\xi$.
Under this setup, in case~$(i)$ of Lemma~\ref{lemma:fastOR_QET} we have either $\sum_{i=0}^{x-1}\frac{x_i}{x} \ge \frac{1}{2}+\frac{\epsilon+2\xi}{4}$ with probability at least $1-\frac{\delta}{m+1}-C_\xi e^{-\alpha n}$, or $\sum_{i=0}^{x-1}\frac{x_i}{x} < \frac{1}{2}+\frac{\epsilon}{4}$ with probability at most $\frac{\delta}{m+1}+C_\xi e^{-\alpha n}$. 
Applying Lemma~\ref{lemma:eigenvalue_disc_density}, we find
\begin{align}
\begin{split}
    &\tr[(\tilde{\rho}\otimes\ket{0^y}\bra{0^y})\Pi_j]
    \\
    &\ge
    \mathrm{Pr}\left[\sum_{i=0}^{x-1}\frac{x_i}{x}\ge \frac{1}{2}+\frac{\epsilon+2\xi}{4}\right] 
    P^{\text{ED}}\left(\sum_{i=0}^{x-1}\frac{x_i}{x}\ge \frac{1}{2}+\frac{\epsilon+2\xi}{4}\right)^2
    \\
    &\ge
    1-3\delta/(m+1)-C_\xi e^{-\alpha n},
\end{split}
\end{align}
where $\sum_i\frac{x_i}{x}$ is an eigenvalue of the operator $\Lambda$ [Eq.~\eqref{eq:amp_projector}].
On the other hand, in case $(ii)$, we have either $\sum_{i=0}^{x-1}\frac{x_i}{x}\ge \frac{1}{2}+\frac{\epsilon+2\xi}{4}$ with probability at most $\frac{\delta}{m+1}+C_\xi e^{-\alpha n}$, or $\sum_{i=0}^{x-1}\frac{x_i}{x}< \frac{1}{2}+\frac{\epsilon}{4}$ with probability at least $1-\frac{\delta}{m+1}-C_\xi e^{-\alpha n}$, which leads to
\begin{align}
\begin{split}
    &\tr[(\tilde{\rho}\otimes\ket{0^y}\bra{0^y})\Pi_j]
    \\
    &\le 
    \mathrm{Pr}\left[\sum_{i=0}^{x-1}\frac{x_i}{x}\ge \frac{1}{2}+\frac{\epsilon+2\xi}{4}\right] 
    + 
    P^{\text{ED}}\left(\sum_{i=0}^{x-1}\frac{x_i}{x}< \frac{1}{2}+\frac{\epsilon}{4}\right)^2
    \\
    &\le 
    2\delta/(m+1) + C_\xi e^{-\alpha n}.
\end{split}
\end{align}

\subsection{Gap amplification of OR lemma}
\label{app:amplify_OR_gap}

Lemma~\ref{lemma:fastOR_QET} with the $m+1$ projectors given by Eq.~\eqref{eq:amp_proj} states that in case $(i)$, the algorithm accepts with probability at least $P_1=(1-3\delta/(m+1)-C_\xi e^{-\alpha n})^2/4-\zeta$, and in case $(ii)$ the algorithm accepts with probability at most $P_2=10\delta+ 5(m+1)C_\xi e^{-\alpha n}+\zeta$. 
To increase the probability of successfully distinguishing the two cases, we amplify the gap $P_1-P_2$ by repeating the test $K$ times. 
If we set $\delta = 1/184$ and $\zeta = 1/32$, then the original gap $P_1-P_2\ge1/8-\calO(e^{-\alpha n})$ is
\begin{equation}
\label{eq:P12}
\begin{aligned}
    P_1 - P_2 
    &= \frac{(1-3\delta/(m+1))^2-C_\xi e^{-\alpha n}}{4}-\zeta 
    - (10\delta+ 5(m+1)C_\xi e^{-\alpha n}+\zeta)
    \\
    &\ge \frac{1}{8} - \calO(m e^{-\alpha n}).
\end{aligned}
\end{equation}
We assume this is a positive quantity of $\calO(1)$.

We let $X_i = 1$ if the $i^\mathrm{th}$ repetition of the test accepts, and $X_i = 0$ if it does not.  Given the count of acceptances $\sum_{i=1}^{K}X_i$, we conclude that there is a broken constraint if $\sum_{i=1}^K X_i \geq K\frac{P_1+P_2}{2}$, and all the constraints are satisfied if $\sum_{i=1}^K X_i \leq K\frac{P_1+P_2}{2}$. 
Making use of the Chernoff bound [Eq.~\eqref{eq:chernoff}],
we find the success probabilities
\begin{equation}
\begin{aligned}
    \mathrm{Pr}\left(\sum_{i=1}^K X_i\geq K \frac{P_1+P_2}{2}\right)
    \ge
    1 - e^{-(P_1-P_2)^2 K/2}
\end{aligned}
\end{equation}
in case $(i)$, and
\begin{equation}
\begin{aligned}
    \mathrm{Pr}\left(\sum_{i=1}^K X_i\leq K\frac{P_1+P_2}{2}\right)
    \ge
    1 - e^{-(P_1-P_2)^2 K/2}
\end{aligned}
\end{equation}
in case $(ii)$. 
With $K = \calO\big(\ln\frac{\log_2m}{\tilde{\delta}}\big)$, the test yields correct outcome with probability at least $1-\calO\big[\big(\tilde{\delta}/\log_2m\big)^{1-\calO(m e^{-\alpha n})}\big] = 1-\tilde{\calO}\big(\tilde{\delta}^{1-\calO(m e^{-\alpha n})}/\log_2m\big)$. 


\newpage 
\bibliographystyle{quantum}
\bibliography{quantum_refs.bib}

\end{document}